\documentclass[a4paper,USenglish,cleveref, autoref, thm-restate, pdfa]{lipics-v2021}

\hideLIPIcs  

\graphicspath{{./figures/}}

\newcommand{\psfrage}[1]{{\color{blue}{\sf[PS: #1]}}}
\newcommand{\mlfrage}[1]{{\color{magenta}{\sf[ML: #1]}}}
\renewcommand{\psfrage}[1]{}\renewcommand{\mlfrage}[1]{}

\title{Fast Many-to-Many Routing for Ridesharing with Multiple Pickup and Dropoff Locations} 

\titlerunning{Fast Many-to-Many Routing for Ridesharing with Multiple Pickup and Dropoff Loc.} 

\author{Moritz Laupichler}{Institute of Theoretical Informatics, Algorithm Engineering, Karlsruhe Institute of Technology, Kaiserstraße 12, 76131 Karlsruhe, Germany}{moritz.laupichler@kit.edu}{https://orcid.org/0009-0001-1193-3477}{}

\author{Peter Sanders}{Institute of Theoretical Informatics, Algorithm Engineering, Karlsruhe Institute of Technology, Kaiserstraße 12, 76131 Karlsruhe, Germany}{sanders@kit.edu}{https://orcid.org/0000-0003-3330-9349}{}

\authorrunning{M. Laupichler and P. Sanders} 

\Copyright{Moritz Laupichler and Peter Sanders} 

\ccsdesc[500]{Applied computing~Transportation}
\ccsdesc[300]{Theory of computation~Online algorithms}
\ccsdesc[100]{Theory of computation~Vector / streaming algorithms}
\ccsdesc[300]{Information systems~Geographic information systems}

\keywords{Algorithm Engineering, Route Planning, Ridesharing, Multi-Modal} 

\category{} 

\relatedversion{} 

\nolinenumbers 

\usepackage[noend]{algpseudocode}
\usepackage{algorithm}
\usepackage{xspace}
\usepackage{siunitx}
\usepackage{makecell}
\usepackage{diagbox}
\usepackage{environ}

\algnewcommand\algorithmicforeach{\textbf{for each}}
\algdef{S}[FOR]{ForEach}[1]{\algorithmicforeach\ #1\ \algorithmicdo}

\newcommand{\includeplot}[1]{%
  \includegraphics[]{plots/#1.pdf}%
}

\newcommand{\karri}[0]{KaRRi\xspace}
\newcommand{\parheader}[1]{\vspace{-3mm}\subparagraph*{#1.}}
\newcommand{\mysubsection}[1]{\vspace{-3mm}\subsection{#1}}

\newcommand{\myDisplayMath}[1]{\vspace{-2mm}\[ #1 \]}

\NewEnviron{myAlign}{%
\vspace{-2mm}
\begin{align*}
  \BODY
\end{align*}
}

\newcommand{\setc}[2]{\left \{ #1 \mid #2 \right \}}
\renewcommand{\gets}[0]{:=}

\newcommand{\Gveh}[0]{\ensuremath{G_{\textit{veh}}}}
\newcommand{\Gpsg}[0]{\ensuremath{G_{\textit{psg}}}}

\newcommand{\dist}[0]{\ensuremath{\delta}}
\newcommand{\distveh}[0]{\ensuremath{\delta_{\textit{veh}}}}
\newcommand{\distpsg}[0]{\ensuremath{\delta_{\textit{psg}}}}
\newcommand{\ellveh}[0]{\ensuremath{\ell_{\textit{veh}}}}
\newcommand{\ellpsg}[0]{\ensuremath{\ell_{\textit{psg}}}}

\newcommand{\orig}[0]{\textit{orig}}
\newcommand{\dest}[0]{\textit{dest}}
\newcommand{\Prho}[0]{\ensuremath{P_{\rho}}}
\newcommand{\Drho}[0]{\ensuremath{D_{\rho}}}

\newcommand{\maxPDDist}[0]{\ensuremath{\dist^{\max}_{\textit{PD}}}}

\newcommand{\capacity}[0]{\ensuremath{\textit{cap}}}
\newcommand{\occupancy}[0]{\ensuremath{\textit{o}}}
\newcommand{\numstops}[1]{\ensuremath{k(#1)}}
\newcommand{\numstopsnu}[0]{\numstops{\nu}}

\newcommand{\laststop}[1]{\ensuremath{s_{\numstops{#1}}}}
\newcommand{\laststopnu}[0]{\laststop{\nu}}

\newcommand{\Ord}[0]{\emph{ordinary}\xspace}
\newcommand{\OP}[0]{\emph{ordinary paired}\xspace}
\newcommand{\PBNS}[0]{\emph{pickup before next stop}\xspace}
\newcommand{\PALS}[0]{\emph{pickup after last stop}\xspace}
\newcommand{\DALS}[0]{\emph{dropoff after last stop}\xspace}

\newcommand{\rank}[0]{\ensuremath{\text{rank}}}

\newcommand{\Gup}[0]{\ensuremath{G^{\uparrow}}}
\newcommand{\Gdown}[0]{\ensuremath{G^{\downarrow}}}
\newcommand{\Vup}[0]{\ensuremath{V^\uparrow}}
\newcommand{\Vdown}[0]{\ensuremath{V^\downarrow}}
\newcommand{\Eup}[0]{\ensuremath{E^\uparrow}}
\newcommand{\Edown}[0]{\ensuremath{E^\downarrow}}
\newcommand{\dup}[0]{\ensuremath{\delta^{\uparrow}}}
\newcommand{\ddown}[0]{\ensuremath{\delta^{\downarrow}}}

\newcommand{\tdist}[0]{\ensuremath{\delta}}

\newcommand{\tddown}[0]{\ddown}

\newcommand{\treq}[0]{\ensuremath{t_{\textit{req}}}}

\newcommand{\tdep}[0]{\ensuremath{t_{\textit{dep}}}}

\newcommand{\tstopmin}[0]{\ensuremath{t_{\textit{stop}}^{\textit{min}}}}
\newcommand{\twaitmax}[0]{\ensuremath{t_{\textit{wait}}^{\textit{max}}}}
\newcommand{\ttripmax}[0]{\ensuremath{t_{\textit{trip}}^{\textit{max}}}}
\newcommand{\tservmin}[0]{\ensuremath{t_{\textit{serv}}^{\textit{min}}}}
\newcommand{\tservmax}[0]{\ensuremath{t_{\textit{serv}}^{\textit{max}}}}
\newcommand{\tdepmin}[0]{\ensuremath{t_{\textit{dep}}^{\textit{min}}}}
\newcommand{\tdepmax}[0]{\ensuremath{t_{\textit{dep}}^{\textit{max}}}}
\newcommand{\tarrmin}[0]{\ensuremath{t_{\textit{arr}}^{\textit{min}}}}
\newcommand{\tarrmax}[0]{\ensuremath{t_{\textit{arr}}^{\textit{max}}}}
\newcommand{\tvehwait}[0]{\ensuremath{t_{\textit{wait}}^{\textit{veh}}}}
\newcommand{\tdeppickup}[0]{\tdep}

\newcommand{\tarrminprime}[0]{\ensuremath{\tarrmin{}'}}
\newcommand{\tdepminprime}[0]{\ensuremath{\tdepmin{}'}}

\newcommand{\remleeway}[0]{\ensuremath{\lambda_{\textit{res}}}}

\newcommand{\tdetour}[0]{\ensuremath{t_{\textit{detour}}}}

\newcommand{\ttrip}[0]{\ensuremath{t_{\textit{trip}}}}

\newcommand{\ttripplus}[0]{\ensuremath{t_{\textit{trip}}^{+}}}
\newcommand{\tride}[0]{\ensuremath{t_{\textit{ride}}}}
\newcommand{\twalk}[0]{\ensuremath{t_{\textit{walk}}}}

\newcommand{\ctripvio}[0]{\ensuremath{c_{\textit{trip}}^{vio}}}

\newcommand{\cwaitvio}[0]{\ensuremath{c_{\textit{wait}}^{vio}}}

\newcommand{\tripweight}[0]{\ensuremath{\tau}}
\newcommand{\walkweight}[0]{\ensuremath{\omega}}
\newcommand{\gammatrip}[0]{\ensuremath{\gamma_{\textit{trip}}}}
\newcommand{\gammawait}[0]{\ensuremath{\gamma_{\textit{wait}}}}

\newcommand{\initpdetour}[0]{\ensuremath{\Delta_{p}^{\textit{init}}}}
\newcommand{\initddetour}[0]{\ensuremath{\Delta_{d}^{\textit{init}}}}
\newcommand{\resdetour}[1]{\ensuremath{\Delta_{#1}^{\textit{res}}}}

\newcommand{\cmaxglobal}[0]{\ensuremath{\hat{c}}} 
\newcommand{\cmax}[0]{\ensuremath{c_{\max}}}
\newcommand{\cmin}[0]{\ensuremath{c_{\min}}}

\newcommand{\Bs}[0]{\ensuremath{B^{\uparrow}}}
\newcommand{\Bt}[0]{\ensuremath{B^{\downarrow}}}

\newcommand{\Blast}[0]{\ensuremath{B^{\uparrow}_{\textit{last}}}}

\newcommand{\open}[0]{\ensuremath{\textit{open}}}
\newcommand{\closed}[0]{\ensuremath{\textit{closed}}}

\newcommand{\bestp}[0]{\ensuremath{p^\ast}}
\newcommand{\bestd}[0]{\ensuremath{d^\ast}}

\newcommand{\bestPALSins}[0]{\ensuremath{\iota_{\text{pals}}^{\ast}}}

\newcommand{\deltacmax}[0]{\ensuremath{\Delta_{\textit{c}}^{\max}}}

\newcommand{\deltadetour}[0]{\ensuremath{\Delta_{\textit{detour}}}}
\newcommand{\deltattrip}[0]{\ensuremath{\Delta_{\textit{trip}}}}
\newcommand{\deltawaitvio}[0]{\ensuremath{\Delta_{\textit{wait}}^{\textit{vio}}}}
\newcommand{\deltatripvio}[0]{\ensuremath{\Delta_{\textit{trip}}^{\textit{vio}}}}

\newcommand{\BerlinOne}[0]{\texttt{Berlin-1pct}\xspace}
\newcommand{\BerlinTen}[0]{\texttt{Berlin-10pct}\xspace}
\newcommand{\RuhrOne}[0]{\texttt{Ruhr-1pct}\xspace}
\newcommand{\RuhrTen}[0]{\texttt{Ruhr-10pct}\xspace}
\newcommand{\ShortBerlinOne}[0]{\texttt{B-1\%}\xspace}
\newcommand{\ShortBerlinTen}[0]{\texttt{B-10\%}\xspace}
\newcommand{\ShortRuhrOne}[0]{\texttt{R-1\%}\xspace}
\newcommand{\ShortRuhrTen}[0]{\texttt{R-10\%}\xspace}

\newcommand{\s}[0]{\si{\second}}
\newcommand{\ms}[0]{\si{\milli\second}}
\newcommand{\mus}[0]{\si{\micro\second}}

\newcommand{\numpickups}[0]{\ensuremath{N_{p}}}
\newcommand{\numdropoffs}[0]{\ensuremath{N_{d}}}

\newcolumntype{R}{>{$}r<{$}} 

\EventEditors{John Q. Open and Joan R. Access}
\EventNoEds{2}
\EventLongTitle{42nd Conference on Very Important Topics (CVIT 2016)}
\EventShortTitle{CVIT 2016}
\EventAcronym{CVIT}
\EventYear{2016}
\EventDate{December 24--27, 2016}
\EventLocation{Little Whinging, United Kingdom}
\EventLogo{}
\SeriesVolume{42}
\ArticleNo{23}

\begin{document}

\maketitle

\begin{abstract}
We introduce KaRRi, an improved
algorithm for scheduling a fleet of shared
vehicles as it is used by services like UberXShare and
Lyft Shared.  We speed up the basic online algorithm that
looks for all possible insertions of a new
customer into a set of existing routes, we
generalize the objective function, and efficiently
support a large number of possible pick-up and
drop-off locations. This lays an algorithmic foundation
for ridesharing systems with higher
vehicle occupancy -- enabling greatly reduced cost and ecological impact
at comparable service quality.
We find that our algorithm computes assignments between vehicles and riders several times faster than a previous state-of-the-art approach.
Further, we observe that allowing meeting points for vehicles and riders can reduce the operating cost of vehicle fleets by up to $15\%$ while also reducing passenger wait and trip times. 
\end{abstract}
\newpage
\setcounter{page}{1}


\section{Introduction}
\label{sec:Introduction}
Current transportation systems are largely based on a
combination of individual transport (often
with heavy, polluting cars that consume
a lot of 
energy and space) and public
transportation that is often slow, inconvenient,
and underdeveloped.  Ridesharing
systems that intelligently control large fleets of
taxi-like vehicles have the potential to offer an
alternative that is more convenient than public
transportation and more economical and ecological
than individually used cars~\cite{8317926, Song2021,YU2017141}. This is particularly
promising if these vehicles have electrical
propulsion and autonomous piloting. However, current
such systems do not deliver on these promises as the effectively
usable capacity of the vehicles is quite small,
even threatening to \emph{increase} rather than
decrease the total number of driven car-kilometers~\cite{Morency2007}.
A main problem for larger capacity ridesharing vehicles is that picking up and dropping off
 customers introduces large delays for other passengers of the vehicle. 

This paper lays algorithmic groundwork for a
better integration of ridesharing fleets into a
multi-modal transportation system.  We focus
on the question of how local transportation (e.g.,
walking, bicycles or scooters) can be used to
reach a pickup or dropoff location (\emph{PD-location}) that causes
less delay for a vehicle, may be shared
with other customers, and may alleviate concerns of privacy for riders~\cite{STIGLIC201536,7549045}.  (A next, related step will
be an analysis of how public transportation like trains or
express buses may be used to cross large
distances faster, and more
economically/ecologically -- overall resulting in
an effective use of the hierarchy individual
transportation, ridesharing, and public
transit).

Our starting point is the LOUD system by Buchhold
et al. \cite{buchhold2021fast} that comprises
an online dispatching system for large ridesharing fleets.
It uses one-to-many routing based on bucket
contraction hierarchies (BCHs) \cite{knopp2007computing,geisberger2012exact} to efficiently
find the best insertion of a new customer into the
current schedule of a vehicle. This is a crucial step
for handling large fleets in real time and computing realistic
simulations of such systems in transportation research.

We introduce the \karri (\underline{Ka}rlsruhe \underline{R}apid \underline{Ri}desharing) algorithm that extends LOUD with the possibility of performing the pickup and dropoff of a passenger not at fixed locations but at any location in the vicinity of the passenger's origin and destination.
The algorithm computes optimal assignments of passengers to vehicles including locations for the pickup and dropoff.
We adapt LOUD's objective function to this new scenario by incorporating passenger wait times, trip times, and overheads for individual transportation to the pickup location and from the dropoff location.

Finding not only the best vehicle for a request but an optimal combination of a vehicle, a pickup location, and a dropoff location leads to a much larger number of possible assignments. 
To determine the best assignment, we need to solve a number of many-to-many routing problems between vehicle locations and \emph{all} possible PD-locations.
We use BCH queries to address this issue and propose novel speedup techniques both for general purpose bucket based queries and for the specific case of localized sources or targets.
We find that these techniques are also applicable for faster routing in the case of a single pickup and dropoff. 

Our experimental evaluation uses realistic data sets to evaluate the efficiency of these measures.
We find that our implementation is several times faster than LOUD in the case of a single pickup and dropoff.
For multiple PD-locations, our routing techniques are up to three orders of magnitude faster than a na\"ive extension of LOUD's techniques.
We also give first indications that allowing multiple PD-locations can reduce the operating costs of a taxi fleet by up to $15\%$ without increasing passenger wait times or trip times .  
A closer investigation of possible effects on the transport system is left to future work likely in cooperation with application experts.

\parheader{Paper Overview}
After a more detailed problem statement in \cref{sec:Problem_Statement} we introduce basic notation and techniques in \cref{sec:Preliminaries}.
We define the formal foundation for our cost model in the presence of multiple PD-locations in~\cref{sec:conceptual_changes}.
\Cref{sec:the_algorithm} gives an overview on the \karri algorithm while~\cref{sec:ordinary_op_and_pbns_insertions,sec:pickup_after_last_stop_insertions,sec:dropoff_after_last_stop_insertions} address individual challenges. 
In \cref{sec:experimental_evaluation}, we evaluate our approach experimentally.


\section{Problem Statement}
\label{sec:Problem_Statement}
This section describes and gives the formal foundations for the dynamic ridesharing problem considered by our approach.

\parheader{Road Network}
We consider a \emph{road network} to be a graph $G=(V_G,E_G)$ where edges represent road segments and vertices represent intersections.
Every edge $e = (v,w) \in E_G$ has a travel time $\ell(e)=\ell(v,w)$.
We denote the \emph{shortest path distance} (i.e. travel time) from a vertex $v$ to a vertex $w$ by $\dist(v,w)$.
Our algorithm uses two separate road networks $\Gveh$ and $\Gpsg$ with associated $\ellveh$, $\distveh$, $\ellpsg$, and $\distpsg$ to represent parts of the same road network accessible to vehicles and to pedestrians, respectively.  
Note that we only consider walking but $\Gpsg$ can also represent other modes of transportation, e.g. cycling.
A passenger can board or alight a vehicle at a location $v$ if $v$ is accessible in both networks, i.e. $v \in V_{veh} \cap V_{psg}$.

\parheader{Vehicle, Stop}
Our algorithm has access to a fleet $F$ of \emph{vehicles}.
Each vehicle $\nu = (l_i, \capacity, \tservmin, \tservmax)$ has an initial location $l_i$, a seating capacity $\capacity$ and a service time interval $[\tservmin, \tservmax)$.
The current \emph{route} $R(\nu) = \langle s_0(\nu), \dots, s_{\numstops{\nu}}(\nu) \rangle$ of a vehicle $\nu$ is a sequence of \emph{stops} scheduled for the vehicle.
The vehicle's current location $l_c(\nu)$ is always somewhere between its previous (or current) stop $s_0(\nu)$ and its next stop $s_1(\nu)$.
Thus, $\numstops{\nu} = |R(\nu)| - 1$ is the number of stops that the vehicle yet has to visit.
Each stop $s$ is mapped to a vertex $l(s) \in V$ in the graph.
Abusing notation, we may write $s_i$ instead of $s_i(\nu)$ and only $s_i$ instead of $l(s_i)$.
At each stop, a vehicle picks up and/or drops off one or more passengers, stopping for a minimum time of $\tstopmin$ which is a model parameter. 
We denote the occupancy of a vehicle between stops $s_i$ and $s_{i+1}$ by $\occupancy(s_i)$.
We update each vehicle's route as new stops are introduced for newly assigned passengers. 
For each stop $s \in R(\nu)$ we maintain the earliest possible arrival time $\tarrmin(s)$ and departure time $\tdepmin(s)$ according to the current schedule.

\parheader{Request}
In our scenario, the dispatcher receives ride requests and immediately assigns them to vehicles.
A request $r=(\orig, \dest, \treq)$ has an origin location $\orig \in V$, a destination location $\dest \in V$ and a time $\treq$ at which the request is issued. 
We do not allow pre-booking, i.e. the request time is also the earliest possible departure time. 

\parheader{Pickup (Location), Dropoff (Location)}
A possible pickup location (\emph{pickup} for short) is a location $v \in V_{psg} \cap V_{veh}$ that is reachable from $\orig(r)$ in $\Gpsg$ within a time radius $\rho$.
Similarly, a possible dropoff location (\emph{dropoff}) is a location $v \in V_{psg} \cap V_{veh}$ from which $\dest(r)$ can be reached in $\Gpsg$ within $\rho$.
The sets of pickups and dropoffs for $r$ and a radius $\rho$ are denoted by $\Prho(r)$ and $\Drho(r)$.
Let $\numpickups(r) = |\Prho(r)|$ and $\numdropoffs(r) = |\Drho(r)|$.
We collectively refer to pickups and dropoffs as \emph{PD-locations}.
We call a pair of pickup and dropoff a \emph{PD-pair}.
The radius $\rho$ is a model parameter.

\parheader{Insertion}
The goal of the dispatcher is to find an insertion of a pickup and dropoff of each request $r$ into any vehicle's route s.t. the cost of that insertion according to a cost function is minimized. 
We formalize an insertion as a tuple $(r, p, d, \nu, i, j)$ indicating that vehicle $\nu$ picks up request $r$ at pickup location $p \in \Prho(r)$ immediately after stop $s_i(\nu)$ and drops $r$ off at dropoff location $d \in \Drho(r)$ immediately after stop $s_j(\nu)$ with $0 \le i \le j \le \numstops{\nu}$ .

\mysubsection{Cost Function and Constraints}
\label{subsec:cost_function}
The cost of an insertion $\iota = (r, p, d, \nu, i, j)$ considers the added vehicle operation time $\tdetour(\iota)$ of $\nu$, the trip time $\ttrip(\iota)$ of $r$, the sum of increased trip times $\ttripplus(\iota)$ of existing passengers of $\nu$, and the walking time $\twalk(\iota)$ (we defer the exact definitions of these terms to~\cref{sec:conceptual_changes}).

We consider a number of constraints for eligible insertions put forward in~\cite{buchhold2021fast}.
After the insertion, the following must hold:
First, the \emph{occupancy} of $\nu$ must never be exceeded by an insertion.
Second, the vehicle must still reach its last stop before the \emph{end of its service time}.
Third, every passenger already assigned to $\nu$ must still be picked up at their pickup stop within a \emph{maximum wait time} $\twaitmax$.
Fourth, every passenger $\hat{r}$ already assigned to $\nu$ must still arrive at their destination within a \emph{maximum trip time} $\ttripmax(\hat{r}) = \alpha \cdot \distveh(\orig(\hat{r}), \dest(\hat{r})) + \beta$.
The values $\twaitmax$, $\alpha$ and $\beta$ are model parameters.

All four constraints are hard constraints wrt. requests already assigned to $\nu$.
If $\iota$ breaks a hard constraint, we set the cost to $\infty$.
For the request $r$ to be inserted, we treat the wait time and trip time constraints as soft constraints, i.e. violating them leads to cost penalties.
Assume, the passenger is picked up at $p$ at time $\tdep$.
We define the cost penalties as
\begin{myAlign}
  \cwaitvio(\iota) &= \gammawait \cdot \max\{\tdep - \treq(r) - \twaitmax, 0 \} \\
  \ctripvio(\iota) &= \gammatrip \cdot \max\{ \ttrip(\iota) - \ttripmax(r), 0 \}
\end{myAlign}
with model parameters $\gammawait$ and $\gammatrip$ that scale the severity of the penalties.

For the total insertion cost, we use a linear combination of the vehicle detour times, passenger trip times, walking times, and soft constraint violation penalties: 
\myDisplayMath{
  c(\iota) = \tdetour(\iota) + \tripweight \cdot (\ttrip(\iota) + \ttripplus(\iota)) + \walkweight \cdot \twalk(\iota) + \cwaitvio(\iota) + \ctripvio(\iota)
}

Note that we base our cost function on the one used in the LOUD algorithm~\cite{buchhold2021fast}.
However, the original cost function does not consider passenger trip times or walking times.
We weight the importance of these times with the model parameters $\tripweight$ and $\walkweight$.
Our cost function is equivalent to LOUD's if $\tripweight = \walkweight = 0$.


\section{Preliminaries}
\label{sec:Preliminaries}
In this section, we describe several algorithms for the computation of shortest paths that are being used in this work.
Furthermore, we summarize the LOUD algorithm for dynamic ridesharing~\cite{buchhold2021fast} that serves as the basis of our work.

\mysubsection{Shortest Path Algorithms}
\label{subsec:shortest_path_algorithms}
In the following, we explain a number of algorithms that compute different variants of shortest path queries on road networks.

\parheader{Dijkstra's Shortest Path Algorithm}
\label{par:dijkstras_shortest_path_algorithm}
\emph{Dijkstra's shortest path algorithm}~\cite{dijkstra1959note} computes the shortest path from a source $s \in V$ to all other vertices in a weighted graph $G=(V,E,\ell)$.

The algorithm stores a distance label $\tdist(s,v)$ for every $v \in V$.
An addressable priority queue (PQ) $Q$ with $\text{key}(v) = \tdist(s,v)$ contains active vertices. 
Initially, $Q \gets \{ s \}$, $\tdist(s,s) \gets 0$ and $\tdist(s,v) \gets \infty$ for $v \ne s$.
The algorithm proceeds by extracting the vertex with the smallest distance label from $Q$ and \emph{settling} it. 
To settle $u \in V$, each outgoing edge $(u,v) \in E$ is \emph{relaxed}.
The relaxation of $e=(u,v)$ tries to improve the distance label $\tdist(s,v)$ with $\tdist(s,u) + \ell(e)$.
If the distance is improved, $v$ is inserted into $Q$.
The algorithm stops when $Q$ becomes empty.

\parheader{Contraction Hierarchies}
\label{par:contraction_hierarchies}
\emph{Contraction Hierarchies (CHs)}~\cite{geisberger2012exact} are a speed-up technique for shortest path computations in road networks that exploits the hierarchical nature of road networks.
A CH is constructed in a pre-processing phase. 
Then, shortest path queries can be computed on the CH using restricted Dijkstra searches.

To construct a CH, all vertices in a road network $G=(V,E)$ are ordered heuristically by their importance or \emph{rank}~\cite{geisberger2012exact}.
Vertices are contracted in the order of increasing rank.
The contraction of $v \in V$ temporarily removes $v$ from the graph.
To preserve shortest paths, a \emph{shortcut edge} $(u,w)$ is created if $(u,v,w) \in E^2$ is the only shortest path between $u$ and $w$.

Let $E^+$ contain all original edges $E$ as well as all shortcut edges.
The graph $G^+=(V, E^+)$ constitutes the CH.
The length $\ell^+(e)$ of a shortcut edge $e$ is the sum of the lengths of replaced original edges while $\dist^+$ is the according distance function.
For the query phase, we partition $E^+$ into \emph{up-edges} $\Eup = \setc{(u,v) \in E^+}{\rank(u) < \rank(v)}$ and \emph{down-edges} $\Edown = \setc{(u,v) \in E^+}{\rank(u) > \rank(v)}$.
We define an \emph{upwards search graph} $\Gup \gets (V, \Eup)$ and a \emph{downwards search graph} $\Gdown \gets (V, \Edown)$.
The distance functions $\dup$ and $\ddown$ represent $\dist^+$ constrained to $\Gup$ and $\Gdown$.
The \emph{upwards CH search space} $\Gup_v = (\Vup_v, \Eup_v)$ rooted at a vertex $v \in V$ contains all vertices $\Vup_v \subseteq V$ that can be reached from $v$ in $\Eup$ and the according edges $\Eup_v \subseteq \Eup$.
The \emph{reverse downwards CH search space} $\Gdown_v = (\Vdown_v, \Edown_v)$ rooted at a vertex $v \in V$ contains all vertices $\Vdown_v \subseteq V$ from which $v$ is reachable in $\Edown$ and the according edges $\Edown_v \subseteq \Edown$.

For any two vertices $s,t \in V$, it can be shown that there is a shortest path from $s$ to $t$ that is an \emph{up-down path} in the CH, i.e. consists of only up-edges followed by only down-edges~\cite{geisberger2012exact}.
A \emph{CH-query} from a source $s \in V$ to a target $t \in V$ runs a forward Dijkstra search from $s$ in $\Gup$ and a reverse Dijkstra search from $t$ in $\Gdown$.
Whenever the searches meet, they find an up-down-path from $s$ to $t$, eventually finding a shortest path.
The query can stop once the radius of either Dijkstra search exceeds the best previously found distance from $s$ to $t$.

\parheader{Bucket Contraction Hierarchy Searches}
\label{par:bucket_contraction_hierarchy_searches}
\emph{Bucket Contraction Hierarchy (BCH)}~\cite{knopp2007computing,geisberger2012exact} searches find all shortest path distances from a source $s \in V$ to a set of targets $T \subseteq V$ in a road network $G=(V,E)$.
A CH $G^+$ of $G$ is used as the basis of the algorithm.

The idea is to construct a \emph{(target) bucket} $\Bt(v)$ at each vertex $v \in V$.
For each target $t \in T$, a reverse search in $\Gdown$ is run that adds an entry $(t, \ddown(v,t))$ to $\Bt(v)$ for every settled $v \in V$.
Then, a forward search from $s$ in $\Gup$ can compute tentative shortest path distances as $\dup(s,v) + \ddown(v,t)$ for every bucket entry $(t, \ddown(v,t)) \in \Bt(v)$ at every settled vertex $v \in \Vdown_s$.

BCH searches can analogously compute the distances from a set of sources to a single target.
In that case, we speak of \emph{source buckets} $\Bs(v)$ for every $v \in V$.

The advantage of BCH searches is that the search space of each source and each target is only traversed once, either to compute bucket entries or to scan bucket entries.
However, storing the bucket entries requires more memory than individual point-to-point CH queries.

\parheader{Bundled Searches}
\label{par:bundled_searches}
Dijkstra-based shortest path algorithms for multiple sources can make use of \emph{bundled searches} where the searches for $k$ sources are advanced simultaneously.
A bundled search maintains $k$ tentative distance labels at each vertex.
The search is rooted at each of the $k$ sources. 
Initially, the $j$-th distance label at the $j$-th source is set to $0$, and all other $kn-k$ distance labels are set to $\infty$.
As usual, vertices are settled by relaxing each outgoing edge.
When the search relaxes an edge $(u,v) \in E$, it tries to update all $k$ distance labels at $v$.

A bundled relaxation can be more cache efficient than $k$ individual relaxations as all $k$ distances are stored in consecutive memory.
However, the relaxation of $(u,v) \in E$ may perform unproductive work if not all $k$ searches have reached $u$ yet.
Thus, bundling is effective if all $k$ searches relax largely the same edges.
The value of $k$ is a tuning parameter.

The concept of bundled searches was first introduced for Dijkstra searches used for the computation of arc-flags under the name \emph{centralized searches}~\cite{hilger2009fast}.
Since then, bundled searches have been used in a number of Dijkstra-based many-to-many shortest path algorithms~\cite{bauer2010sharc,yanagisawa2010multi,delling2011faster,delling2013phast,delling2017customizable}.
More recently, the idea has been extended to point-to-point queries in CHs~\cite{buchhold2019real}.

\parheader{Instruction-Level Parallelism in Bundled Searches}
Additionally, single-instruction multiple-data (SIMD) parallelism can be utilized for bundled searches~\cite{buchhold2019real}.
Modern CPUs provide special vector registers and instructions that can store and manipulate multiple data items simultaneously.
We can vectorize the computations needed during edge relaxations s.t. $k$ computations are performed at the same time using a single vector instruction.
SIMD instructions can substantially speed up bundled searches~\cite{buchhold2019real}.

\mysubsection{LOUD}
\label{subsec:LOUD}
Our algorithm is based on the dynamic ridesharing dispatching algorithm \emph{LOUD}~\cite{buchhold2021fast}. 

Given a fleet of vehicles and routes, the online algorithm matches incoming ridesharing requests to vehicles.
For each request, a feasible insertion of the request's origin $o$ and destination $d$ into a vehicle's route is found s.t. the detour of the vehicle is minimized.

\parheader{Elliptic Pruning}
To compute the costs of possible insertions, the algorithm requires the distances between existing vehicle stops and $o$ and $d$.
LOUD computes these distances using BCHs with bucket entries for each vehicle stop and queries run from $o$ and $d$.

We refer to these BCH searches as \emph{elliptic BCH searches} as they utilize a pruning technique for these buckets called \emph{elliptic pruning}:
Each insertion is subject to the same soft and hard constraints that we describe in~\cref{subsec:cost_function}.
The wait time and trip time hard constraints of passengers already assigned to a vehicle $\nu \in F$ define a \emph{leeway} $\lambda(s_i,s_{i+1})$, i.e. a maximum permissible detour, between each pair of consecutive stops $(s_i,s_{i+1}) \in R(\nu)$. 
Any detour that exceeds $\lambda(s_i,s_{i+1})$ breaks some hard constraint and is infeasible. 
The leeway $\lambda(s_i,s_{i+1})$ defines a detour ellipse that contains all vertices at which a pickup or dropoff may be made between $s_i$ and $s_{i+1}$ without breaking a hard constraint.
Thus, bucket entries for $s_i$ and $s_{i+1}$ only need to be generated at vertices within the ellipse.
Elliptic pruning vastly reduces the number of bucket entries that need to be scanned by the BCH searches and limits the number of candidate vehicles for insertions~\cite{buchhold2021fast}.

\parheader{Last Stop Distances}
LOUD also allows the insertion of the origin and/or destination after the last stop of a vehicle's route.
Here, elliptic pruning is not applicable since the leeway of any vehicle is unbounded after the last stop.
Instead, LOUD uses reverse Dijkstra queries in the road network rooted at $o$ or $d$ to find the distances from last stops to $o$ or $d$.
These Dijkstra queries, particularly for distances from last stops to the destination of a request, constitute a significant part (at least $60\%$ and up to more than $90\%$) of the total runtime of LOUD.


\section{Conceptual Changes for Multiple Pickup and Dropoff Locations}
\label{sec:conceptual_changes}
We observe that introducing passenger movement for pickup and dropoff locations requires a careful consideration of its effects on vehicle detours and passenger trips.
In the following, we describe these effects in detail, leading us to the formal foundation of our cost function.

Remark that we ignore two special cases in this section for the sake of simplicity:
First, we assume that the pickup and dropoff for an insertion $\iota = (r, p, d, \nu, i, j)$ are inserted after the next stop of $\nu$, i.e. $i \ge 1$.
Second, we do not consider the possibility of $p$ or $d$ coinciding with existing stops, i.e. we assume $p \ne l(s_i)$ and $d \ne l(s_j)$.
We ignore these cases as they would lead to bloated definitions.
However, with knowledge of the vehicle's current location $l_c(\nu)$ and the location $l(s)$ of each stop $s \in R(\nu)$, the ignored cases could be integrated into the following definitions in a straight forward manner.

\subsection{Walking Time and Walking to the Destination}
\label{subsec:walking_time_and_walking_to_the_destination}
The walking time of a regular insertion $\iota=(r,p,d,\nu,i,j)$ is simply $\twalk(\iota) \gets \distpsg(\orig,p) + \distpsg(d, \dest)$.

We allow each passenger $r$ to walk from their origin to their destination without ever boarding a vehicle.
This requires a manner of pseudo-insertion $\iota_{\text{psg}}$ where the passenger is matched to no vehicle at all. 
Then, the cost of $\iota_{\text{psg}}$ depends only on the walking distance $\twalk(\iota_{\text{psg}}) = \ttrip(\iota_{\text{psg}}) = \distpsg(\orig, \dest)$.
The pseudo-insertion affects no vehicle operation times or trip times of other passengers. 
We ignore the wait time soft constraint since the passenger does not wait for a vehicle.
In effect, the total cost is $c(\iota_{\text{psg}}) = \tripweight \cdot \ttrip(\iota_{\text{psg}}) + \walkweight \cdot \twalk(\iota_{\text{psg}}) + \ctripvio(\iota_{\text{psg}})$.
We explicitly allow pseudo-insertions for any distance $\distpsg(\orig,\dest)$, i.e. the distance does not have to be found within the radius $\rho$ around $\orig$ or $\dest$.
Instead, we use a CH query in the passenger graph to find $\distpsg(\orig,\dest)$.
The cost $c(\iota_{\text{psg}})$ can serve as a first upper bound on the cost of any insertion.

\subsection{Vehicles Waiting for Passengers}
\label{subsec:vehicles_waiting_for_passengers}
In the traditional dynamic ridesharing model, a request $r=(\orig, \dest, \treq)$ always waits to be picked up by the vehicle at the origin location $\orig$.
The request is issued at time $\treq$ and the vehicle $\nu$ matched to the request can start making its way to the pickup location at the earliest at $\treq$. 
If the pickup location is $\orig$ this means that the vehicle will always arrive at the pickup location later than the passenger, i.e. only the passenger can wait for the vehicle, not the other way around.
In that case, the time a vehicle is stopped at any stop is always exactly $\tstopmin$ (unless the vehicle is idling, i.e. it currently has no other stops to get to).
Therefore, for each stop $s$, the scheduled arrival time can be inferred from the scheduled departure time as $\tarrmin(s) = \tdepmin(s) - \tstopmin$.
In the traditional dynamic ridesharing model, it suffices to store the departure time at each stop for the full vehicle schedule~\cite{buchhold2021fast}.
The shortest path distance between two consecutive stops $s_i$ and $s_{i+1}$ is then $\distveh(s_i, s_{i+1}) = \tdepmin(s_{i+1}) - \tstopmin - \tdepmin(s_i)$. 

Now, consider an insertion $\iota = (r, p, d, \nu, i, j)$ with $p \ne \orig$.
Then, both the vehicle and the passenger have to travel to $p$ starting at $\tdepmin(s_i)$ and $\treq$, respectively.
Consequently, the vehicle can arrive at $p$ earlier than the passenger, precisely if $\tdepmin(s_i) + \distveh(s_i, p) < \treq + \distpsg(\orig, p)$.
In that case, the vehicle needs to wait for the passenger at $p$. 
An inserted stop $s$ at $p$ can then take longer than $\tstopmin$. 
We, therefore, define the actual earliest possible departure time at each pickup as the maximum of the earliest possible departure time of the vehicle and that of the passenger at the pickup.
The vehicle $\nu$ can depart from $p$ at the earliest after a stop time of $\tstopmin$ and the passenger can depart as soon as they arrive at $p$ leading us to
\myDisplayMath{
  \tdeppickup(\iota) = \max(\tdepmin(s_i) + \distveh(s_i, p) + \tstopmin, \treq + \distpsg(\orig, p))\text{.}
}
We later use $\tdeppickup(\iota)$ in the definitions of the vehicle detour, passenger wait time and passenger trip time needed for the cost function (see~\cref{subsec:cost_function}). 
In particular, the wait times of a vehicle or of a passenger contribute to the vehicle operation times and passenger trip times.

Whereas in the traditional model we can infer $\tarrmin(s)$ from $\tdepmin(s)$ for a stop $s$, we have to now explicitly store $\tarrmin(s)$ and $\tdepmin(s)$.
The distance between a pair of consecutive stops $(s_i, s_{i+1})$ can then be derived as $\distveh(s_i, s_{i+1}) = \tarrmin(s_{i+1}) - \tdepmin(s_i)$.

We denote the scheduled wait time of a vehicle at stop $s$ with $\tvehwait(s) = \tdepmin(s) - \tarrmin(s) - \tstopmin$.

\subsection{Added Vehicle Operation Time and Passenger Trip Times}
\label{subsec:times_for_costs_definitions}
In the following, we define the added vehicle operation time $\tdetour(\iota)$, the trip time for the new passenger $\ttrip(\iota)$, as well as the sum of added trip times for existing passengers $\ttripplus(\iota)$ for an insertion $\iota = (r, p, d, \nu, i, j)$.
Vehicle wait times as explained in~\cref{subsec:vehicles_waiting_for_passengers} have an effect on all of these times.

In order to explain this effect, we first focus on the added vehicle operation time caused by the insertion $\iota$.
Conceptually, the added vehicle operation time is the time difference between the vehicle's arrival time at the last scheduled stop after the insertion and the arrival time at the last scheduled stop before the insertion.

\parheader{Initial Detours}
We start by explaining added vehicle operation times in a scenario without vehicle waiting times, i.e. $\tvehwait(s) = 0$ for all $s \in R(\nu)$.
In this case, the added vehicle operation time for performing a pickup at $p$ and a dropoff at $d$ is simply equal to the detour that vehicle $\nu$ has to take in its route after stop $s_i$ and stop $s_j$, respectively.
We call these detours the initial pickup detour and initial dropoff detour.

\begin{definition} 
\label{def:initial_detour}
The initial pickup (dropoff) detour for an insertion $\iota = (r,p,d,\nu,i,j)$ is the detour that results from the vehicle $\nu$ first driving to $p$ ($d$) after stop $s_i$ ($s_j$) instead of driving to $s_{i+1}$ ($s_{j+1}$) directly.\\ 
Formally, we define the initial pickup detour as 
\myDisplayMath{ 
  \initpdetour(\iota) \gets \begin{cases} 
    \tdeppickup(\iota) - \tdepmin(s_i) & \text{ if } i = j \\
    \tdeppickup(\iota) - \tdepmin(s_i) + \distveh(p, s_{i+1}) - \distveh(s_i, s_{i+1}) & \text{ if } i \ne j 
  \end{cases} 
}
We define the initial dropoff detour as
\myDisplayMath{ 
  \initddetour(\iota) \gets \begin{cases} 
    \distveh(p, d) + \tstopmin & \text{ if } i = j \text{ and } j = \numstops{\nu} \\
    \distveh(s_j, d) + \tstopmin & \text{ if } i \ne j \text{ and } j = \numstops{\nu} \\
    \distveh(p, d) + \tstopmin + \distveh(d, s_{j+1}) - \distveh(s_j, s_{j+1}) & \text{ if } i = j \text{ and } j \ne \numstops{\nu} \\
    \distveh(s_j, d) + \tstopmin + \distveh(d, s_{j+1}) - \distveh(s_j, s_{j+1}) & \text{ if } i \ne j \text{ and } j \ne \numstops{\nu}
  \end{cases}
}
\end{definition}

Note that if $i = j$, the vehicle has to make a combined detour to go from stop $s_i$ to $p$, then to $d$ and finally to $s_{i+1}$. 
In that case, only the leg between $s_i$ and $s_{i+1}$ of the existing route $R(\nu)$ is replaced by the combined detour.
The definition of initial pickup and dropoff detour account for this by subtracting the distance between $s_i$ and $s_{i+1}$ only once in the dropoff detour and not in the pickup detour if $i = j$.

\parheader{Residual Detours, Added Vehicle Operation Time}
Without vehicle wait times, the departure time $\tdeppickup(\iota)$ at pickup $p$ is the arrival time of $\nu$ at $p$ plus the minimum stopping time $\tstopmin$, so $\tdeppickup(\iota) = \tdepmin(s_i) + \dist(s_i, p) + \tstopmin$ (see~\cref{subsec:vehicles_waiting_for_passengers}).
The first effect that vehicle wait times can have on the detour is the fact that $\tdeppickup(\iota)$ can depend on the passenger if they arrive at the pickup later than the vehicle.
As defined in~\cref{subsec:vehicles_waiting_for_passengers}, with passenger movement we have $\tdeppickup(\iota) = \max(\tdepmin(s_i) + \dist(s_i,p) + \tstopmin, \treq(r) + \distpsg(\orig, p))$. 
Therefore, in a scenario with vehicle waiting times, the initial pickup detour $\initpdetour(\iota)$ can be larger than without them.

Furthermore, existing vehicle wait times at stops in $R(\nu)$ can also factor into the added vehicle operation time.
Assume $j + 1 < a < \numstops{\nu}$, $\tvehwait(s_{a}) > 0$ and $\tvehwait(s_b) = 0$ for $b \ne a$.
Let $\tarrminprime(s, \iota)$ and $\tdepminprime(s, \iota)$ describe the scheduled arrival time and departure time at stop $s \in R(\nu)$ after we perform insertion $\iota$.
After the insertion, vehicle $\nu$ will arrive at stop $s_{j+1}$ with a delay of $\resdetour{j + 1}(\iota) = \initpdetour(\iota) + \initddetour(\iota)$ compared to before the insertion, so $\tarrminprime(s_{j+1}) = \tarrmin(s_{j+1}) + \resdetour{j + 1}(\iota)$.
This delay in arrival times is propagated forward through the route until stop $s_{a}$, so for the delay arriving at $s_a$ we have $\resdetour{a}(\iota) = \resdetour{j + 1}(\iota)$. 
However, before the insertion, the vehicle has to wait for a duration of $\tvehwait(s_a)$ at stop $s_a$ for a passenger that is scheduled to be picked up here.
Since the vehicle now arrives at $s_a$ later, the waiting time may decrease.
If the delay at $s_a$ is larger than the previous wait time at $s_a$, the vehicle may even now arrive at $s_a$ later than the passenger.

Effectively, the vehicle now uses the time that was spent waiting for the passenger at $s_a$ to perform part of the detour needed for the pickup and dropoff of $\iota$.
Thus, the vehicle wait times at stops after the pickup or dropoff of an insertion act as buffers to the added vehicle operation time.
In our example, we have \begin{align*}\tarrminprime(s_a, \iota) &= \tarrmin(s_a) + \resdetour{a}(\iota) \text{ but} \\ \tdepminprime(s_a, \iota) &= \tdepmin(s_a) + \resdetour{a + 1}(\iota) \text{ where } \resdetour{a + 1}(\iota) = \max(\resdetour{a}(\iota) - \tvehwait(s_a), 0)\end{align*}
Note that the vehicle wait time at $s_a$ reduces the delay for the departure time at $s_a$ and also the arrival times at all following stops.
Formally, for every stop $s_b$ with $a < b$ we get $\resdetour{b}(\iota) = \resdetour{a + 1}(\iota)$ so $\tarrminprime(s_b, \iota) = \tarrmin(s_b) + \resdetour{a + 1}(\iota)$.

This leads us to the definition of residual detours which describe the actual delay that an insertion $\iota$ causes for the arrival time at any existing stop of the route $R(\nu)$.
It takes into account the possibility that a vehicle wait time may exist at any stop that reduces the delay for all later stops.

\begin{definition} 
\label{def:residual_detour}
  The residual detour $\resdetour{a}(\iota)$ for an insertion $\iota = (r, p, d, \nu, i, j)$ at stop $s_a \in R(\nu)$ describes how much later the vehicle $\nu$ will arrive at stop $s_a$ after the insertion is performed. 
  Formally, we define it as
  \[ \resdetour{a}(\iota) \gets \begin{cases} 
      0 & \text{ if } a \le i \\
      \initpdetour(\iota) & \text{ if } i + 1 = a \le j \\
      \max(\resdetour{j}(\iota) - \tvehwait(s_j), 0) + \initddetour(\iota) & \text{ if } a = j + 1 \text{ and } i \ne j \\
      \initpdetour(\iota) + \initddetour(\iota) & \text{ if } a = j + 1 \text{ and } i = j \\
      \max(\resdetour{a - 1}(\iota) - \tvehwait(s_{a-1}), 0) & \text{ otherwise }
  \end{cases}\]
\end{definition}

Residual detours allow us to easily define the new arrival times $\tarrminprime$ and departure times $\tdepminprime$ after an insertion $\iota$ at each stop $s_a \in R(\nu)$ as
\begin{align*}
  \tarrminprime(s_a, \iota) &= \tarrmin(s_a) + \resdetour{a}(\iota) \text{ and}\\
  \tdepminprime(s_a, \iota) &= \max(\tarrmin{}'(s_a) + \tstopmin, \tdepmin(s_a)) \text{.}
\end{align*}

Then, the added vehicle operation time can be defined as:\\

\begin{definition} 
\label{def:added_vehicle_operation_time}
  The added vehicle operation time $\tdetour(\iota)$ caused by an insertion $\iota = (r, p, d, \nu, i, j)$ is defined as
\[ \tdetour(\iota) \gets  \begin{cases}
                            \initpdetour(\iota)+ \initddetour(\iota) & \text{ if } i = j = \numstops{\nu} \\
                            \tarrminprime(\laststop{\nu}, \iota) - \tarrmin(\laststop{\nu}) + \initddetour(\iota) & \text{ if } i < j = \numstops{\nu} \\
                            \tarrminprime(\laststop{\nu}, \iota) - \tarrmin(\laststop{\nu}) & \text{ otherwise}   
                          \end{cases}\] 
\end{definition}

\parheader{Trip Time, Added Trip Time For Existing Passengers}
The previous definitions also enable us to define the trip time $\ttrip(\iota)$ of the passenger associated with the request $r$. 
The trip time for $r$ is simply the time between the request time $\treq(r)$ and the scheduled arrival of the passenger at their destination $\dest(r)$.
Chronologically, a trip starts with the passenger moving to the pickup point $p$ and possibly waiting for the vehicle to arrive at $p$.
This is followed by the actual ride in the vehicle from $p$ to $d$.
Finally, the passenger moves from the dropoff $d$ to their destination.

\begin{definition} 
\label{def:trip_time}
  The trip time $\ttrip(\iota)$ for an insertion $\iota = (r, p, d, \nu, i, j)$ is defined as 
  \[ \ttrip(\iota) \gets (\tdeppickup(\iota) - \treq(r)) + \tride(\iota) + \distpsg(d, \dest(r))  \]
  where 
  \[ \tride(\iota) \gets \begin{cases}
    \distveh(p, d) & \quad\text{ if } i = j \\
    \distveh(p, s_{i+1}) + (\tdepminprime(s_j, \iota) - \tarrminprime(s_{i+1}, \iota)) + \distveh(s_j, d) & \quad\text{ if } i \ne j 
  \end{cases} \] 
\end{definition}

The final contribution to the total cost of an insertion is the trip time that is added for passengers that are already assigned to vehicle $\nu$.
Consider a request $r'$ that has previously been assigned to vehicle $\nu$.
Assume $r'$ is picked up by $\nu$ at stop $s_{i'}$ and dropped off at stop $s_{j'}$.
By definition of residual detours, performing the insertion $\iota$ will delay the arrival of the vehicle at $s_{j'}$ by $\resdetour{j'}(\iota)$.
Therefore, the trip time of $r'$ increases by that same amount of time.

\begin{definition} 
\label{def:added_trip_time_for_existing_passengers}
  Let $\numdropoffs(s)$ be the number of dropoffs currently scheduled to be performed at stop $s \in R(\nu)$ for a vehicle $\nu \in F$.
  The combined added trip time for existing passengers $\ttripplus(\iota)$ caused by an insertion $\iota = (r, p, d, \nu, i, j)$ is defined as
  \[ \ttripplus(\iota) \gets \sum_{a = i + 1}^{\numstops{\nu}} \numdropoffs(s_a) \cdot \resdetour{a}(\iota) \]
\end{definition}


\section{Algorithm Overview}
\label{sec:the_algorithm}
We introduce the \emph{\karri} (\underline{Ka}rlsruhe \underline{R}apid \underline{Ri}desharing) algorithm that efficiently answers ridesharing requests with multiple PD-locations using fast many-to-many routing.  
\begin{figure}[tb]
  \centering
  \begin{tabular}{ccl}
    \hline
    \makecell[c]{Ordinary} & \makecell[c]{$0 < i < j < \numstopsnu$} & \makecell[l]{\rule{0pt}{1.1cm}\rule[-.15cm]{0pt}{.15cm}\includegraphics[height=0.9cm,page=2]{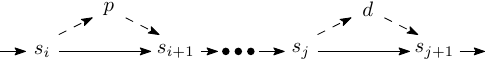}} \\
    \hline
    \makecell[c]{Ordinary\\Paired (OP)} & \makecell[c]{$0 < i = j < \numstopsnu$} & \makecell[l]{\rule{0pt}{1.1cm}\rule[-.15cm]{0pt}{.15cm}\includegraphics[height=0.9cm,page=3]{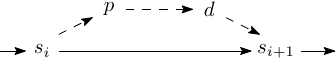}} \\
    \hline
    \makecell[c]{Pickup Before\\Next Stop (PBNS)} & \makecell[c]{$0 = i \le j < \numstopsnu$} & \makecell[l]{\rule{0pt}{1.1cm}\rule[-.15cm]{0pt}{.15cm}\includegraphics[height=0.9cm,page=4]{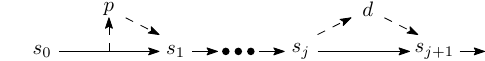}} \\
    \hline
    \makecell[c]{Pickup After\\Last Stop (PALS)} & \makecell[c]{$0 \le i = j = \numstopsnu$} & \makecell[l]{\rule{0pt}{1.1cm}\rule[-.15cm]{0pt}{.15cm}\includegraphics[height=0.9cm,page=5]{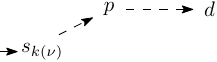}} \\
    \hline
    \makecell[c]{Dropoff After\\Last Stop (DALS)} & \makecell[c]{$0 < i < j = \numstopsnu$} & \makecell[l]{\rule{0pt}{1.1cm}\rule[-.15cm]{0pt}{.15cm}\includegraphics[height=0.9cm,page=6]{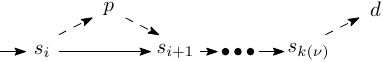}} \\
    \hline
  \end{tabular}
  \caption{Insertion types.
  Shows characterization of each type based on the pickup and dropoff insertion points $i$ and $j$ of an insertion $\iota=(r,p,d,\nu,i,j)$. 
  Illustrations depict the current route of $\nu$ (solid arrows) with stops $s \in R(\nu)$ as well as the detours to and from $p$ and $d$ (dashed lines).}
  \label{fig:insertion_types}
\end{figure}
The \karri algorithm dynamically accepts requests and finds an insertion for each request that has optimal cost according to the cost function and current system state.

For a request $r$, the algorithm first finds the possible PD-locations in a walking radius $\rho$ around the origin and destination using bounded Dijkstra searches.
Then, the algorithm evaluates all insertions in the order of types illustrated in~\cref{fig:insertion_types}.
For each insertion, \karri computes the cost according to the cost function (see~\cref{subsec:cost_function}).
The insertion with the smallest cost $\iota^\ast$ is repeatedly updated and eventually returned.

Since we consider sets of possible PD-locations, the number of potential insertions becomes the main challenge of the algorithm.
In particular, we face the issue of computing the shortest paths between existing vehicle stops and each PD-location to filter out infeasible insertions and to determine the cost of the remaining candidate insertions.
In the following sections, we describe the bundling and filtering methods that we employ to limit the running time of the required many-to-many shortest path queries for each insertion type.

\section{Ordinary, Ordinary Paired, and Pickup Before Next Stop Insertions}
\label{sec:ordinary_op_and_pbns_insertions}
This section is concerned with \Ord, \OP, and \PBNS insertions (see~\cref{fig:insertion_types}).
In all three insertion types mentioned, both the pickup $p$ and the dropoff $d$ are inserted between two existing stops of the route $R(\nu)$.
To compute the cost of any insertion of one of these types, we need to know the distances between existing vehicle stops and the PD-locations. 
BCH searches with elliptic pruning (see~\cref{subsec:LOUD}) have been shown to efficiently compute these distances~\cite{buchhold2021fast}.
We call these \textit{elliptic BCH searches}.
Additionally, cost calculations for paired insertions require the PD-distance $\distveh(p,d)$.
In this section, we explain how we extend the required distance queries for multiple PD-locations.

\subsection{Elliptic BCH Searches}
\label{subsec:elliptic_bch_searches}
We can extend elliptic BCH queries to multiple PD-locations by simply repeating the queries for each PD-location.
Even with elliptic pruning, this can lead to impractical running times for large numbers of PD-locations, though.
Therefore, we supplement elliptic BCH searches with two techniques for better scalability to larger numbers of PD-locations. 
We describe these techniques only for pickups but they work analogously for dropoffs.

\parheader{Elliptic BCH Searches with Sorted Buckets}
\label{par:elliptic_bch_searches_with_sorted_buckets}
First, we propose using \emph{sorted buckets} to reduce the number of bucket entries scanned by BCH queries.
We explain the principle only for source buckets but it can be analogously applied for target buckets.

Recall that the constraints for existing passengers of a vehicle $\nu$ define a leeway $\lambda(s_i, s_{i+1})$ for the detour between any two consecutive vehicle stops $s_i$ and $s_{i+1}$ of $\nu$ (see~\cref{subsec:LOUD}).
When an elliptic BCH search to a pickup $p$ scans a source bucket entry $(s, \dup(s_i,v)) \in \Bs(v)$, the tentative distance $\dup(s_i,v) + \tddown(v,p)$ can only lead to an insertion that holds all hard constraints if $\dup(s_i,v) + \tddown(v,p) \le \lambda(s_i, s_{i+1})$.
This means the entry is only relevant for $p$ if $\tddown(v,p) \le \lambda(s_i, s_{i+1}) - \dup(s_i,v)$.
We call $\remleeway(s_i, v) \gets \lambda(s_i,s_{i+1}) - \dup(s_i,v)$ the \emph{remaining leeway} of a source bucket entry $(s_i, \dup(s_i,v))$ at vertex $v$.

We sort the entries of each source bucket at each vertex $v$ by their remaining leeway in decreasing order.
Then, an elliptic BCH search to a pickup $p$ can stop scanning the entries at $v$ once an entry $(s, \dup(s,v))$ is scanned for which $\tddown(v,p) > \remleeway(s,v)$.
In this case, we have $\tddown(v,p) > \remleeway(s,v) \ge \remleeway(s',v)$ for any subsequent entry $(s', \dup(s',v))$, so the remaining entries cannot lead to any insertions that adhere to the hard constraints.
Maintaining the order of each bucket comprises a time overhead when inserting bucket entries.
However, since bucket sizes are small, this overhead is limited.
Note that sorted buckets can also be applied in the case of only a single pickup.

\parheader{Bundled Elliptic BCH Searches}
\label{par:bundled_elliptic_bch_searches}
Second, we employ \emph{bundled elliptic BCH searches} that exploit the locality of pickups.

Like any bundled search, a bundled elliptic BCH search is rooted at $k$ pickups and updates $k$ distances with each edge relaxation (see~\cref{par:bundled_searches}).
Additionally, we can bundle bucket entry scans.
Whenever a bucket entry for a stop $s$ is scanned, the search tries to improve upon any of the $k$ tentative distances between $s$ and any of the $k$ pickups.  
We can effectively bundle the edge relaxations and bucket entry scans of elliptic BCH searches because the localized pickups share similar CH search spaces.  
Moreover, we can use vectorized instructions to parallelize both edge relaxations and bucket entry scans.
At the same time, elliptic pruning and sorted buckets can still be applied.
To our knowledge, our algorithm is the first to explicitly use bundled BCH searches. 
The idea follows from the bundled CH searches used in~\cite{buchhold2019real}.

\subsection{PD-Distance Searches}
\label{subsec:pd_distance_searches}
Computing the PD-distances, i.e. the distances between pickups and dropoffs, is a many-to-many shortest path problem where the set of sources and the set of targets are localized.

Our algorithm uses a BCH approach to address this problem.
We generate bucket entries for all dropoffs in their reverse CH search spaces.
Then, we run queries in the upward CH search graph rooted at each pickup to find the PD-distances using the dropoff bucket entries.
We propose two methods to improve these BCH searches.

Firstly, let $\maxPDDist$ be an  an upper bound on all PD-distances.
Then, we only have to generate and scan bucket entries in a radius of $\maxPDDist$. 
We use
\myDisplayMath{
  \maxPDDist \gets \max_{p \in \Prho} \distveh(p, \orig) + \distveh(\orig,\dest) + \max_{d\in \Drho} \distveh(\dest, d)\text{.}
}
We can compute $\distveh(p,\orig)$ for all $p \in \Prho$ and $\distveh(\dest, d)$ for all $d \in \Drho$ using two local Dijkstra searches rooted at $\orig$ and $\dest$, respectively.
We obtain $\distveh(\orig,\dest)$ with a single preliminary CH query. 

Secondly, we can once again use bundled BCH searches.
More specifically, we can generate bucket entries for batches of $k$ dropoffs and then run queries for batches of $k$ pickups where $k$ is a configuration parameter. 
Again, bundled PD-distance searches utilize the locality of pickups and dropoffs and allow us to employ SIMD parallelism.

\subsection{Enumerating Ordinary, Ordinary Paired, and Pickup Before Next Stop Insertions}
\label{subsec:enumerating_ordinary_op_and_pbns_insertions}
After running our elliptic BCH queries and PD-distance searches, we know all distances that are required for ordinary and \OP insertions.
We enumerate the insertions $\iota=(r,p,d,\nu,i,j)$ with $0 < i \le j < \numstopsnu$ for a set of candidate vehicles found by the elliptic BCH queries~\cite{buchhold2021fast}.
We compute the cost $c(\iota)$ for each insertion and update $\iota^\ast$ to $\iota$ if $c(\iota) < c(\iota^\ast)$.

To compute the cost of a \PBNS (PBNS) insertion $\iota=(r,p,d,\nu,0,j)$, we lack knowledge of the distance $\distveh(l_c(\nu), p)$ from the vehicle's current location $l_c(\nu)$ to the pickup $p$.
LOUD suggests filtering PBNS insertions based on a lower bound on the cost of $\iota$ obtained from considering $\distveh(s_0,p)$ instead of $\distveh(s_0,l_c(\nu)) + \distveh(l_c(\nu), p)$ for the pickup detour~\cite{buchhold2021fast}.
This filter discards almost all PBNS insertions. 
For the remaining PBNS insertions, we compute the missing distances $\distveh(l_c(\nu), p)$ with a bucket based approach.
We generate source bucket entries for the current location of every affected vehicle and run bundled queries from the pickups.
The average number of such queries per request is less than $0.5$.

\section{Pickup After Last Stop Insertions}
\label{sec:pickup_after_last_stop_insertions}
In this section, we consider \PALS (PALS) insertions.
The main challenge of PALS insertions is the computation of the distances from last stops to pickups.
The authors of LOUD find that elliptic pruning is not applicable for the computation of these distances~\cite{buchhold2021fast}. 
Instead, LOUD uses a reverse Dijkstra search rooted at $\orig$ that is stopped early when the search can no longer find an insertion better than the best known one.
For multiple pickups, we can compute the required distances by analogously running reverse Dijkstra searches for each pickup. 
These Dijkstra searches may also be bundled to exploit the locality of the pickups.

However, even with a single pickup, this Dijkstra search takes up a significant part of the running time of the LOUD algorithm. 
Thus, for a large number of pickups, we expect infeasible running times.
In this section, we introduce two new BCH based approaches for the computation of last stop distances.
For the rest of this section, let $\cmaxglobal$ denote an upper bound on the best known insertion cost (initially $\cmaxglobal \gets c(\iota^\ast)$).

\parheader{Reformulation of Cost Function for PALS Insertions}
Note that the cost of any PALS insertion $\iota=(r,p,d,\nu,\numstopsnu, \numstopsnu)$ is fully characterized by the pickup $p$, the PD-distance $\distveh(p,d)$, the walking distance $\distpsg(d, \dest)$, the departure time $\tdepmin(\laststopnu)$ of $\nu$ at $\laststopnu$, and the last stop distance $\distveh(\laststopnu,p)$. 
Thus, we can write the cost of $\iota$ as 
\myDisplayMath{
  c(\iota) = c'(r,p,\distveh(p,d), \distpsg(d,\dest), \tdepmin(\laststopnu), \distveh(\laststopnu, p))\text{.}
}

\subsection{Last Stop BCH Searches for PALS}
\label{subsec:last_stop_bch_searches_for_PALS}
Even though elliptic pruning is not applicable, we can still employ a BCH search approach for distances from last stops to pickups.
For this, we maintain a \emph{last stop bucket} $\Blast(v)$ for every $v \in V$.
For every last stop $\laststopnu$, we generate an entry $(\laststopnu, \dup(\laststopnu,v)) \in \Blast(v)$ at each vertex $v$ in the upward CH search space rooted at $\laststopnu$.
Then, for every pickup $p \in \Prho$, we run an \emph{individual (last stop) BCH query} in $\Gdown_p$ that scans the last stop bucket at each settled vertex to compute the shortest path distances from last stops to $p$.
When the search scans an entry $(\laststopnu, \dup(\laststopnu,v)) \in \Blast(v)$, it tries to improve the tentative distance $\tdist(\laststopnu, p)$ with $\dup(\laststopnu, v) + \ddown(v,p)$.
Eventually, the shortest distance $\distveh(\laststopnu, p)$ will be found for every last stop $\laststopnu$.

Whenever the last stop of a vehicle changes, we run two forward CH searches to remove the bucket entries of the old last stop and insert new entries for the new last stop.
We stop the search early when the current distance no longer admits a PALS insertion with cost smaller than the best known cost.
Furthermore, we can bundle the BCH queries and use SIMD parallelism in a similar manner to bundled elliptic BCH searches (see~\cref{par:bundled_elliptic_bch_searches}).

\parheader{Cost Pruning of Bucket Scans using Sorted Buckets}
\label{par:cost_pruning_of_bucket_scans_using_sorted_buckets}
A remaining issue of this approach is the size of the last stop buckets. 
Without elliptic pruning, buckets contain many more entries, especially at vertices that have a high rank in the CH.
Therefore, the queries have to scan large numbers of bucket entries, rendering the last stop BCH approach ineffective.

The future work section of~\cite{buchhold2021fast} suggests sorting the entries within each last stop bucket to address this issue.
For each $v \in V$, we sort the entries in $\Blast(v)$ by their distance $\dup(\laststopnu, v)$ in increasing order.
Suppose a pickup query rooted at $p \in \Prho$ scans the bucket $\Blast(v)$.
Let 
\[ 
  \cmin(x) \gets c'(r,p,\delta_{\text{pd}}^{\min}, 0, \treq(r), x + \ddown(v,p)) \text{\quad where } \delta_{\text{pd}}^{\min} \gets \min_{p \in \Prho, d \in \Drho} \distveh(p,d) \text{.}
\]

Then, for any entry $e=(\laststopnu, \dup(\laststopnu, v)) \in \Blast(v)$, the value $\cmin(\dup(\laststopnu, v))$ is a vehicle-independent lower bound for the cost of any PALS insertion where a vehicle drives a distance of at least $\dup(\laststopnu,v) + \ddown(v,p)$ to $p$.
Note that $\cmin(x)$ is linear in $x$.
Thus, since the entries of $\Blast(v)$ are sorted by $\ddown(\laststopnu, v)$, we can stop the bucket scan if $\cmin(\dup(\laststopnu,v)) > \cmaxglobal$.

\parheader{Updating the Upper Bound Cost}
\label{par:pals_updating_upper_bound}
It is possible to simply use $c(\iota^\ast)$ for the upper bound cost $\cmaxglobal$ needed for cost pruning.
However, we can also improve $\cmaxglobal$ during the search.
Each tentative distance $\tdist(\laststopnu, p)$ found acts as an upper bound on the actual shortest distance $\distveh(\laststopnu,p)$. 
Thus, whenever the tentative distance $\tdist(\laststopnu, p)$ is updated, we can compute an upper bound
\myDisplayMath{ 
  \cmax = c'(r,p,\distveh(p,\dest),0, \tdepmin(\laststopnu),\tdist(\laststopnu, p))
} 
on the cost of the best PALS insertion with $\nu$ and $p$. 
We update $\cmaxglobal$ to $\cmax$ if $\cmax < \cmaxglobal$.
This technique finds inexact cost upper bounds early which is helpful for the stopping criterion of bucket scans.

\subsection{Collective Last Stop Searches for PALS}
\label{subsec:collective_last_stop_searches_for_PALS}
Finally, we propose a search approach based on the idea that we do not actually need to know the distance between every last stop and every pickup.
If we knew the best PALS insertion $\bestPALSins=(r,p,d,\nu,\numstopsnu,\numstopsnu)$ in advance, we would only need to find $\distveh(\laststopnu, p)$.
Obviously, we do not know $\bestPALSins$ in advance but we find that it is possible to prune the distance queries for individual pickups (or actually PD-pairs) by comparing them to each other.
We introduce a collective BCH query that finds the best PALS insertion $\bestPALSins$ as well as the last stop distance $\distveh(\laststopnu, p)$.
In the following, we explain how labels representing PD-pairs are propagated through the CH search graph and how these labels can be pruned based on label domination.

\parheader{Open and Closed Labels}
\label{par:open_and_closed_labels}
A PD-pair label $(p,d,\ddown(v,p))$ at a vertex $v \in V$ consists of the pickup $p \in \Prho$, dropoff $d \in \Drho$ and downwards distance $\ddown(v,p)$. 
At each vertex $v \in V$, there is a set of \emph{open} labels $\open(v)$ and a set of \emph{closed} labels $\closed(v)$.
An open label is a label that has not been settled yet.
For each open label $l=(p,d,\ddown(v,p))$, we store a lower bound $\cmin(l)$ for the cost of a PALS insertion that can be found for $l$ in $\Gdown_v$
\[
  \cmin(l) \gets c'(r,p,\distveh(p,d),\distpsg(d,\dest),\treq(r),\ddown(v,p))
\]

\parheader{Algorithm Outline}
\label{par:collective_pals_algorithm_outline}
\begin{algorithm}[t]
      \caption{Collective BCH search used to find distances from last stops to pickups.}
      \label{alg:collective_bch}
      \begin{algorithmic}[1]
      \Procedure{\texttt{CollectiveBCH}}{$\Prho$, $\Drho$, $\Gdown = (\Vdown,\Edown)$} : $(\bestp, \bestd)$ and $\distveh(\laststopnu, \bestp)$ \label{alg:collective_bch:main_body}
        \State $Q \gets$ PQ of labels with $key_Q(l) = \cmin(l)$; $\open(v) \gets \emptyset$, $\closed(v) \gets \emptyset$ for $v \in V$
        \State $\cmaxglobal \gets c(\iota^\ast)$
        \ForEach{$(p,d) \in \Prho \times \Drho$}
          \State \texttt{insertLabelAtVertex}($p$, $(p,d,0)$)
        \EndFor
        \While{$Q \ne \emptyset$}
          \State $l \gets Q$.\texttt{deleteMin}()
          \If{$\cmin(l) > \cmaxglobal$} \Return \EndIf
          \State \texttt{settleLabel}($l$)
        \EndWhile
        \EndProcedure

        \Statex \phantom{empty line}

        \Procedure{\texttt{settleLabel}}{$l = (p,d,\ddown(v,p))$} \label{alg:collective_bch:settleLabel}
          \State $\open(v)$.\texttt{remove}($l$); $\closed(v)$.\texttt{insert}($l$) \Comment{mark $l$ closed}

          \smallskip

          \ForEach{$e = (\laststop{\nu}, \dup(\laststop{\nu}, v)) \in \Blast(v)$} \Comment{scan last stop entries at $v$ for $l$}
            \If{$\cmin(l,e) > \cmaxglobal$} \textbf{break} \EndIf
            \If{$\cmax(l,e) < \cmaxglobal$}
              \State $(\bestp, \bestd) \gets (p,d)$; $\cmaxglobal \gets \cmax(l,e)$
            \EndIf
          \EndFor

          \smallskip

          \ForEach{$(u,v) \in \Edown$} \Comment{propagate $l$ to neighbors of $v$}
            \State \texttt{insertLabelAtVertex}($u$, $(p,d,\ell^+(u,v) + \ddown(v,p))$)
          \EndFor

        \EndProcedure

        \Statex \phantom{empty line}

        \Procedure{\texttt{insertLabelAtVertex}}{$v$, $l' = (p,d,\ddown(v,p))$} \label{alg:collective_bch:insertLabelAtVertex}
          \If {$\cmin(l') > \cmaxglobal$} \Return \EndIf
          \ForEach{$l \in \open(v) \cup \closed(v)$}
            \If{$l$ dominates $l'$} \Return \EndIf
          \EndFor
          \ForEach{$l \in \open(v)$}
            \If{$l'$ dominates $l$} $\open(v)$.\texttt{remove}($l$) \EndIf
          \EndFor
          \State $\open(v)$.\texttt{insert}($l'$); $Q$.\texttt{insert}($l'$)
        \EndProcedure
      \end{algorithmic}
\end{algorithm}
We give pseudocode for a collective BCH search in~\cref{alg:collective_bch}.
Our search maintains a priority queue $Q$ that contains all open labels ordered increasingly by $\cmin$.
Initially, at each pickup $p \in \Prho$, an open label $(p,d,0) \in \open(p)$ is created for each $d \in \Drho$.
As long as $Q$ contains a label $l$ with $\cmin(l) \le \cmaxglobal$ for a known upper bound $\cmaxglobal$ on the cost of any insertion, our search proceeds with a next step.
In each step of the search, the label $l \gets \min(Q)$ is removed from $Q$ and settled as described in the following.

\parheader{Settling Open Labels}
\label{par:settling_open_labels_for_collective_PALS}
Settling an open label $l = (p,d,\ddown(v,p))$ consists of three steps:
First, we mark $l$ closed at $v$, i.e. we move $l$ from $\open(v)$ to $\closed(v)$.
Second, we search for a new best insertion by traversing all entries in the last stop bucket $\Blast(v)$.
For each entry $e = (\laststop{\nu}, \dup(\laststop{\nu}, v)) \in \Blast(v)$, we compute an upper bound cost
\[
  \cmax(l,e) \gets c'(r,p,\distveh(p,d),\distpsg(d,\dest),\tdepmin(\laststopnu),\dup(\laststopnu, v) + \ddown(v,p)) \text{.} 
\]
If $\cmax(l,e) < \cmaxglobal$, we mark $\iota=(r,p,d,\nu,\numstopsnu,\numstopsnu)$ as the best known PALS insertion, store the tentative distance $\tdist(\laststopnu, p) = \dup(\laststopnu, v) + \ddown(v, p)$, and update $\cmaxglobal \gets \cmax(l,e)$. 
Note that $\cmax(l,e)$ is the exact cost of the PALS insertion $\iota=(r,p,d,\nu,\numstopsnu,\numstopsnu)$ if $\tdist(\laststopnu, p)$ is a shortest path distance.
Since the BCH search finds shortest up-down paths, we will thus eventually find the best PALS insertion.
As before, we can stop each bucket scan early.
For this purpose, we compute a vehicle-independent cost lower bound $\cmin(l,e)$ s.t. we can stop the search early if $\cmin(l,e) > \cmaxglobal$ using
\[
  \cmin(l,e) \gets c'(r,p,\distveh(p,d), \distpsg(d,\dest), \treq(r), \dup(\laststopnu, v) + \ddown(v,p)) \text{.}
\]
Third, we propagate $l$ to all neighboring vertices of $v$.
For each neighboring vertex $w \in V$ with $(w,v) \in \Gdown$, we create a new open label $l' = (p,d,\ell^+(w,v) + \ddown(v,p))$ at $w$.
Here, we employ cost pruning by discarding $l'$ if the lower bound cost $\cmin(l')$ for this PD-pair and this distance exceeds $\cmaxglobal$.
Furthermore, we may be able prune $l'$ at $v$ if it is dominated by an existing label at $v$ as described in the following.

\parheader{Domination Pruning}
\label{par:domination_pruning_for collective_PALS}
Propagating a label through the entire search space for every PD-pair is too expensive.
However, we find that we can compare labels at the same vertex and prune dominated labels in a technique we call \textit{domination pruning}.
Intuitively, a label $l$ dominates a label $l'$ at a vertex $v$ if we know that any insertion found in $\Gdown_v$ that uses $l'$ has higher costs than the equivalent insertion using $l$.

To formalize this, we first define an upper bound for the cost of a PALS insertion found in $\Gdown_v$ for a label $l$.
Let $l=(p,d,\ddown(v,p))$ be a label at $v$ and $e = (\laststop{\nu}, \dup(\laststop{\nu}, w)) \in \Blast(w)$ a last stop bucket entry at a vertex $w \in \Vdown_v$.
Then, we define an upper bound for the cost of an insertion where $\nu$ drives from $\laststopnu$ to $p$ via $w$ and $v$ as
\myDisplayMath{
  \cmax(l,v,e) \gets c'(r,p,\distveh(p,d),\distpsg(d,\dest),\tdepmin(\laststopnu), \dup(\laststopnu, w) + \ddown(w,v) + \ddown(v, p))
}

With this, we can formally define the domination relation between labels:

\begin{definition} \label{def:pd_pair_domination}
  A PD-pair label $l$ \emph{dominates} another label $l'$ at a vertex $v \in V$ exactly if $\cmax(l,v,e) < \cmax(l',v,e)$ for every $w \in \Vdown_v$ and $e \in \Blast(w)$.
\end{definition}

\begin{theorem}
If a label $l$ dominates another label $l'$ at $v$, we do not need to settle $l'$ at $v$.
\end{theorem}

\begin{proof} \label{proof:pals_domination}
This can be shown by contradiction.
Assume $l=(p,d,\ddown(v,p))$ dominates $l'=(p',d',\ddown(v,p'))$ at $v$.
Further, assume that $\iota=(r,p',d',\nu,\numstopsnu,\numstopsnu)$ is the best PALS insertion.
Let $\pi$ be a shortest path from $\laststop{\nu}$ to $p'$.
Wlog. $\pi$ is an up-down-path in the CH consisting of an upwards prefix $\pi^\uparrow$ and a downwards suffix $\pi^\downarrow$.
If $\pi^\downarrow$ does not contain $v$, then the collective search will not find $\pi$ in $\Gdown_v$, and we do not have to settle $l'$ at $v$.

Otherwise, $\pi^\downarrow = (w, \dots, v, \dots, p')$ with $w \in \Vdown_v$.
Let $e = (\laststop{\nu}, \dup(\laststop{\nu}, w)) \in \Blast(w)$.
Since $\pi$ is a shortest path, we know that
\begin{align*}
  \cmax(l',v,e) =& c'(r,p',\distveh(p',d'),\distpsg(d',\dest),\tdepmin(\laststopnu), \distveh(\laststopnu, p')) \\
  =& c((r,p',d',\nu,\numstopsnu, \numstopsnu)) \text{.}
\end{align*}
However, $l$ dominates $l'$ which means that
\myDisplayMath{
  c((r,p,d,\nu,\numstopsnu, \numstopsnu)) \le \cmax(l,v,e) < \cmax(l',v,e) = c((r,p',d',\nu,\numstopsnu, \numstopsnu))
}

This contradicts $\iota$ being the best PALS insertion. 
Hence, we do not have to settle label $l'$ at $v$ to find the best pair for $\nu$.
\end{proof}

\parheader{Efficiently Computing the Domination Relation}
\label{par:efficiently_computing_pals_domination}
We find that it is not trivial to compute the domination relation efficiently because of the non-linearity of the cost function.

Consider two labels $l_i=(p_i,d_i,\ddown(v,p_i))$ for $i=1,2$ and two PALS insertions $\iota_i=(r,p_i,d_i,\nu,\numstopsnu,\numstopsnu)$ found for these labels in $\Gdown_v$.
Assume the vehicle $\nu$ arrives at $v$ at some time $t = \tdepmin(\laststopnu) + \distveh(\laststopnu, v)$.
If we know $t$, we can determine the cost difference $c(\iota_1) - c(\iota_2)$ irrespective of the actual vehicle $\nu$.
In other words, the difference $\Delta_c(l_1, l_2, t)$ in costs between any insertions that can be found for $l_1$ and $l_2$ in $\Gdown_v$ is a function of $t$.
Then, if $\Delta_c(l_1, l_2, t) < 0$ for all $t \ge 0$, every insertion found in $\Gdown_v$ will be better with $l_1$ than with $l_2$, i.e. $l_1$ dominates $l_2$ .

If the cost function for PALS insertions were to grow linearly with $t$, then $\Delta_c(l_1,l_2,t)$ would be constant wrt. $t$.
In that case, we could simply check for domination using $\Delta_c(l_1,l_2,0) < 0$. 
However, the cost function does not increase linearly with $t$: 
Firstly, the cost is constant wrt. $t$ as long as the vehicle arrives at $p_i$ earlier than the passenger, i.e. $t + \ddown(v,p_i) + \tstopmin \le \treq(r) + \distpsg(\orig, p_i)$ (see~\cref{subsec:vehicles_waiting_for_passengers}).
Secondly, due to the wait time and trip time soft constraints, linear penalty terms are added to the cost function starting at a certain threshold for the wait time and trip time, both of which $t$ contributes to.

The passenger arrival time at $p_i$ and the thresholds for soft constraint penalties differ between labels.
Thus, $\Delta_c(l_1,l_2,t)$ varies with $t$.   
Since we do not know which values of $t$ are possible for insertions found in $\Gdown_v$, we cannot trivially determine whether $l_1$ dominates $l_2$.

\parheader{Approximating the Domination Relation}
\label{par:approximating_pals_domination}
Instead, we under-approximate the domination relation by computing a sufficient precondition.
For this purpose, we find an upper bound $\deltacmax(l_1, l_2) \ge \max_{t \ge 0} \Delta_c(l_1,l_2,t)$ on the difference in insertion costs between any insertion that can be found for $l_1$ and $l_2$ in $\Gdown_v$.
Then, $l_1$ dominates $l_2$ if $\deltacmax(l_1,l_2) < 0$. 

We now explain how to find this upper bound:
The maximum difference in the departure times at $p_1$ and $p_2$ is $\tdepmax(l_1) - \tdepmin(l_2)$ with
\begin{align*}
    \tdepmax(l_1) &\gets \max \{ \ddown(v,p_1) + \tstopmin, \distpsg(p_1, \orig)\} \\
    \tdepmin(l_2) &\gets \ddown(v,p_2) + \tstopmin
\end{align*}
Then, for every insertion found in $\Gdown_v$, the difference in detours between $l_1$ and $l_2$ is bounded by $\deltadetour(l_1,l_2) \gets \tdepmax(l_1) + \distveh(p_1,d_1) - \tdepmin(l_2) - \distveh(p_2,d_2)$.
For the difference in trip times, we first define:
\begin{align*}
    \tarrmax(l_1) &\gets \tdepmax(l_1) + \distveh(p_1,d_1) + \distpsg(d_1, \dest) \\
    \tarrmin(l_2) &\gets \tdepmin(l_2) + \distveh(p_2,d_2) + \distpsg(d_2,\dest)
\end{align*}
Then, the difference in trip times is bounded by $\deltattrip(l_1,l_2) \gets \tarrmax(l_1) - \tarrmin(l_2)$.
We define upper bounds for the difference in penalties for the wait and trip time soft constraints as
\begin{align*}
    \deltawaitvio(l_1,l_2) &\gets \gammawait \max \{ \tdepmax(l_1) - \tdepmin(l_2), 0 \} \\
    \deltatripvio(l_1,l_2) &\gets \gammatrip \max \{ \deltattrip(l_1,l_2), 0 \}
\end{align*}
Note that the differences in detours and trip times are allowed to be negative to express a cost advantage for $l_1$ but the differences in penalties are not.
Even if $\tdepmax(l_1) < \tdepmin(l_2)$ or $\deltattrip(l_1,l_2) < 0$, we may find insertions in $\Gdown_v$ where no penalties apply for either label.
In those cases, the penalty difference has to be zero.
Let $\Delta_{\textit{walk}} (l_1,l_2)$ be the fix difference in walking costs.
Putting it all together, we get
\[ \deltacmax(l_1,l_2) \gets \deltadetour(l_1,l_2) + \tripweight \deltattrip(l_1,l_2) + \walkweight \Delta_{\textit{walk}}(l_1,l_2) + \deltawaitvio(l_1,l_2) + \deltatripvio(l_1,l_2) \]

We can compute $\deltacmax(l_1,l_2)$ in constant time with information that is known at $v$ without looking ahead in the search tree.
Since we under-approximate domination, it is possible that $l_1$ actually dominates $l_2$ but our condition does not hold.
However, we find that our domination criterion still manages to prune the vast majority of labels early.

\parheader{Limitations}
\label{par:limitations_of_collective_pals}
We remark that the insertion $\iota$ found by the collective search is only guaranteed to be the best possible PALS insertion if $\iota$ holds the service time hard constraint.
Since our search ignores the service time constraint, it may return an insertion that breaks the constraint even if there are other PALS insertions that do not.

Therefore, if $\iota$ breaks the service time constraint, we fall back to computing the distances from every last stop to every pickup using individual last stop BCH searches. 
The fallback individual BCH searches can make use of the good cost upper bounds found during the collective search.
We find that this is only necessary in exceedingly rare cases.

\section{Dropoff After Last Stop Insertions}
\label{sec:dropoff_after_last_stop_insertions}
The final insertion type are \DALS (DALS) insertions (see~\cref{fig:insertion_types}).
Similarly to PALS insertions, we face the problem of not knowing the distances from the vehicles' last stops to dropoffs.
We can use the same approaches of Dijkstra searches as well as individual and collective BCH searches to find these distances with some minor differences.

Firstly, cost pruning is less effective than in the PALS case since the lower bounds on costs cannot include the PD-distance. 
Secondly, we cannot update the global cost upper bound $\cmaxglobal$ during the searches in the DALS case as we lack information about the cost of inserting the pickup earlier in the route.
Thirdly, collective BCH searches have some more intricate differences between the PALS and DALS cases. 
We go into more detail about these differences in the rest of this section.

\parheader{Pareto-Best Dropoffs}
\label{par:pareto_best_dropoffs}
\begin{figure}[t!]
  \centering
  \includegraphics[width=0.75\textwidth]{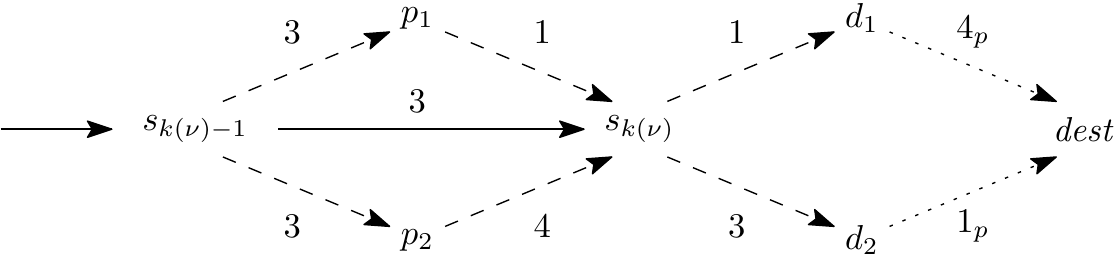}
  \caption{Example \DALS insertions that illustrate why domination between dropoffs is a partial relation.
  Shows existing vehicle route (solid) with current last two stops ($s_{\numstopsnu - 1}, \laststopnu$), two pickups between these stops ($p_1, p_2$) and two dropoffs after the last stop ($d_1, d_2$) with necessary detours (dashed), as well as the walking distances from either dropoff to the destination (dotted).
  Edges are annotated with vehicle travel times or passenger travel times (marked with $p$).}
  \label{fig:partial_dropoff_domination}
\end{figure}

The largest difference of collective searches in the DALS case compared to the PALS case is the goal of the search.
Whereas collective PALS searches always find a single best PD-pair, we cannot find a single best dropoff.
Instead, we find a set of pareto-best dropoffs per vehicle. 

To see why there cannot always be a single best dropoff for a vehicle $\nu \in F$, assume that we would like to decide whether $d_1 \in \Drho$ always leads to better DALS insertions than $d_2 \in \Drho$.
Dropoff $d_1$ is always better than $d_2$ for $\nu$ if $c(\iota_1) < c(\iota_2)$ for any DALS insertions $\iota_1=(r,p,d_1,\nu,i,\numstopsnu)$ and $\iota_2=(r,p,d_2,\nu,i,\numstopsnu)$ for any pickup $p \in \Prho$ and $0 \le i < \numstopsnu$.

Consider the cost difference $c(\iota_1) - c(\iota_2)$.
The vehicle detour and trip time is the same for both insertions up to $\laststopnu$.
Thus, the differences between $\iota_1$ and $\iota_2$ in the detours, trip times, added trip times for existing passengers, walking times, and penalties for the wait time constraint are all entirely independent of the choice of $p$ and $i$. 
However, for different choices of $p$ and $i$, either insertion may or may not violate the trip time soft constraint.
Thus, it is possible that $\iota_1$ is better for some choice of $p$ and $i$ while $\iota_2$ may better for another.

To illustrate this issue, consider the example depicted in~\cref{fig:partial_dropoff_domination}.
Suppose $\tripweight = 1$ and $\gammatrip = 10$.
Further, let the maximum trip time $\tarrmax(r) - \treq(r) = 10$.
Let $\distpsg(\orig, p_1) < 3$ as well as $\distpsg(\orig, p_2) < 3$.

Consider the insertions $\iota_{1i}=(r,p_1,d_i,\nu,\numstopsnu-1,\numstopsnu)$ for $i=1,2$.
Neither insertion violates the trip time soft constraint as $\ttrip(\iota_{11}) = 3 + 1 + 1 + 4 = 9$ and $\ttrip(\iota_{12}) = 3 + 1 + 3 + 1 = 8$.
The difference in detours is $\tdetour(\iota_{11}) - \tdetour(\iota_{12}) = 5 - 7 = -2$, leading to a difference in costs of $c(\iota_{11}) - c(\iota_{12}) = -1$.
Therefore, $d_1$ is the better choice for $p_1$.

Now, consider the insertions $\iota_{2i}=(r,p_2,d_i,\nu,\numstopsnu-1,\numstopsnu)$ for $i=1,2$.
Here, both insertions violate the trip time soft constraint since $\ttrip(\iota_{21}) = 12$ and $\ttrip(\iota_{22}) = 11$.
The difference in detours and trip times add up like before but with the added difference in penalties, we get $c(\iota_{21}) - c(\iota_{22}) = -1 + 10 = 9$.
Thus, $d_2$ is the better choice for $p_2$.

This demonstrates that we may not be able to choose a single best dropoff for DALS insertions. 
Instead, we only know that a dropoff $d_1$ always leads to better insertions than a dropoff $d_2$ if $d_1$ is better for DALS insertions with and without a trip time violation.
Thus, we can only ever obtain a set of pareto-best dropoffs for each vehicle.

\parheader{Algorithm Outline and Partial Domination}
\label{par:collective_DALS_algorithm_outline_and_partial_domination}
Our collective search maintains open and closed dropoff labels of the form $(d, \ddown(v,d))$ at each vertex $v \in V$ with $d \in \Drho$.
Initially, an open label $(d, 0)$ is created at every $d \in \Drho$.
As in the PALS case, we associate a lower bound cost $\cmin(l)$ with each open label $l$.
In each step, the open label $l$ with the smallest lower bound cost is settled by closing the label at its vertex $v$, scanning the last stop bucket at $v$, and propagating $l$ to neighboring vertices in $\Gdown_v$.

We use our cost based stopping criteria for bucket scans and the entire search.
When a label is propagated to a new vertex, we apply domination pruning.
For the domination relation, we test whether a dropoff will always be better than another dropoff with and without trip time penalties: 
\begin{definition}
  Let $l_i=(d_i,\ddown(v,d_i))$ for $i=1,2$ be two dropoff labels at $v \in V$. 
  Let 
  \begin{align*}
    \deltadetour(l_1,l_2) &\gets \ddown(v,d_1) - \ddown(v,d_2) \\
    \deltattrip(l_1,l_2) &\gets \ddown(v,d_1) + \distpsg(d_1, \dest) - \ddown(v,d_1) - \distpsg(d_2, \dest)
  \end{align*}
  Then $d_1$ \emph{dominates} $d_2$ if 
  \begin{enumerate}
    \item $\deltadetour(l_1,l_2) + \tripweight \deltattrip(l_1,l_2) < 0$ and
    \item $\deltadetour(l_1,l_2) + (\tripweight + \gammatrip) \deltattrip(l_1,l_2) < 0$
  \end{enumerate}
\end{definition}

\section{Experimental Evaluation}
\label{sec:experimental_evaluation}
\begin{table}[t]
  \centering
  \caption{Key figures of our benchmark instances.}
  \begin{tabular}{|l|R|R|R|R|}
    \hline
    Instance & \makecell[c]{|V|} & \makecell[c]{|E|} & \makecell[c]{\#\text{veh.}} & \makecell[c]{\#\text{req.}} \\
    \hline
    \BerlinOne & 73689 & 159039 & 1000  & 16569 \\
    \hline
    \BerlinTen & 73689 & 159039 & 10000 & 149185 \\
    \hline
    \RuhrOne & 394049 & 840587 & 3000 & 49707 \\
    \hline
    \RuhrTen & 394049 & 840587 & 30000 & 447555 \\
    \hline
  \end{tabular}
  \label{tab:key_figures_of_instances}
\end{table}
\begin{table}
  \centering
  \caption{Average number of pickups ($\numpickups$) and dropoffs ($\numdropoffs$) for specific values of the walking radius $\rho$ on the \BerlinOne, \BerlinTen, \RuhrOne, and \RuhrTen instances.}
  \begin{tabular}{|r|r|r|r|r|r|r|r|r|}
    \cline{2-9}
    \multicolumn{1}{c|}{} & \multicolumn{2}{c|}{\ShortBerlinOne} & \multicolumn{2}{c|}{\ShortBerlinTen} & \multicolumn{2}{c|}{\ShortRuhrOne} & \multicolumn{2}{c|}{\ShortRuhrTen} \\
    \hline
    \makecell[c]{$\rho$} & \makecell[c]{$\numpickups$} & \makecell[c]{$\numdropoffs$}  & \makecell[c]{$\numpickups$} & \makecell[c]{$\numdropoffs$} & \makecell[c]{$\numpickups$} & \makecell[c]{$\numdropoffs$}  & \makecell[c]{$\numpickups$} & \makecell[c]{$\numdropoffs$} \\
    \hline
    $0\s$ & 1 & 1 & 1 & 1 & 1 & 1 & 1 & 1 \\
    \hline
    $150\s$ & 10 & 10 & 10 & 10 & 12 & 11 & 12 & 11 \\
    \hline
    $300\s$ & 33 & 31 & 33 & 31 & 35 & 33 & 35 & 33 \\
    \hline
    $450\s$ & 68 & 66 & 68 & 67 & 70 & 68 & 70 & 68 \\
    \hline
    $600\s$ & 115 & 113 & 115 & 114 & 114 & 112 & 114 & 112 \\
    \hline
  \end{tabular}
  \label{tab:num_pickups_dropoffs_with_radius}
\end{table}
Our source code is written in C++17 and compiled with GCC 9.4 using \texttt{-O3}.
We check the correctness of our implementation with a rigorous use of assertions (disabled for experiments). 
We conduct our experiments on a machine running Ubuntu 20.04 with 512 GiB of memory and two 16-core Intel Xeon E5-2670 v3 processors at 2.3GHz.
We use 32-bit distance labels and the AVX2 SIMD instruction set with 256-bit registers to compute up to $8$ operations in one vector instruction.

We evaluate our implementation on the \BerlinOne (\ShortBerlinOne), \BerlinTen (\ShortBerlinTen), \RuhrOne (\ShortRuhrOne), and \RuhrTen (\ShortRuhrTen) request sets~\cite{buchhold2021fast} that respectively represent 1\% and 10\% of ridesharing demand in the Berlin and Rhein-Ruhr metropolitan areas on a weekday. 
The request sets for Berlin were artificially generated based on the Open Berlin Scenario~\cite{ziemke2019matsim} for the MATSim transport simulation~\cite{horni2016multi}.
The request sets for the Rhein-Ruhr area were obtained by scaling up the Berlin request sets~\cite{buchhold2021fast}.
The underlying road networks are obtained from OpenStreetMap data\footnote{\url{https://download.geofabrik.de/}}.
We use the known speed limit of each road to determine the travel time of the according edge in the vehicle network.
For the passenger network, we assume a constant walking speed of $4.5$km/h. 
We show the sizes of the networks and request sets in~\cref{tab:key_figures_of_instances}.
We scale the number of pickups $\numpickups$ and dropoffs $\numdropoffs$ by using increasing walking radii $\rho \in \{0\s, 150\s, 300\s, 450\s, 600\s\}$ which lead to the numbers of PD-locations given in~\cref{tab:num_pickups_dropoffs_with_radius}. 
Unless stated otherwise, we run $5$ iterations per combination of algorithm configuration, input, and radius $\rho$ and report the average running times. 

For our model parameters, we choose $\twaitmax = 600\s$, $\tstopmin = 60\s$, $\gammawait = 1$, $\gammatrip = 10$, $\tripweight=1$, and $\walkweight=0$.
We use the open-source library RoutingKit\footnote{\url{https://github.com/RoutingKit/RoutingKit}} to compute the required CHs of the road networks which takes only a few seconds for our instances.
This pre-processing step takes only a few seconds for each graph on any of our input instances.

\subsection{Bundled Searches}
\label{subsec:bundled_searches_experiments}
\begin{figure}[t]
  \centering
  \includeplot{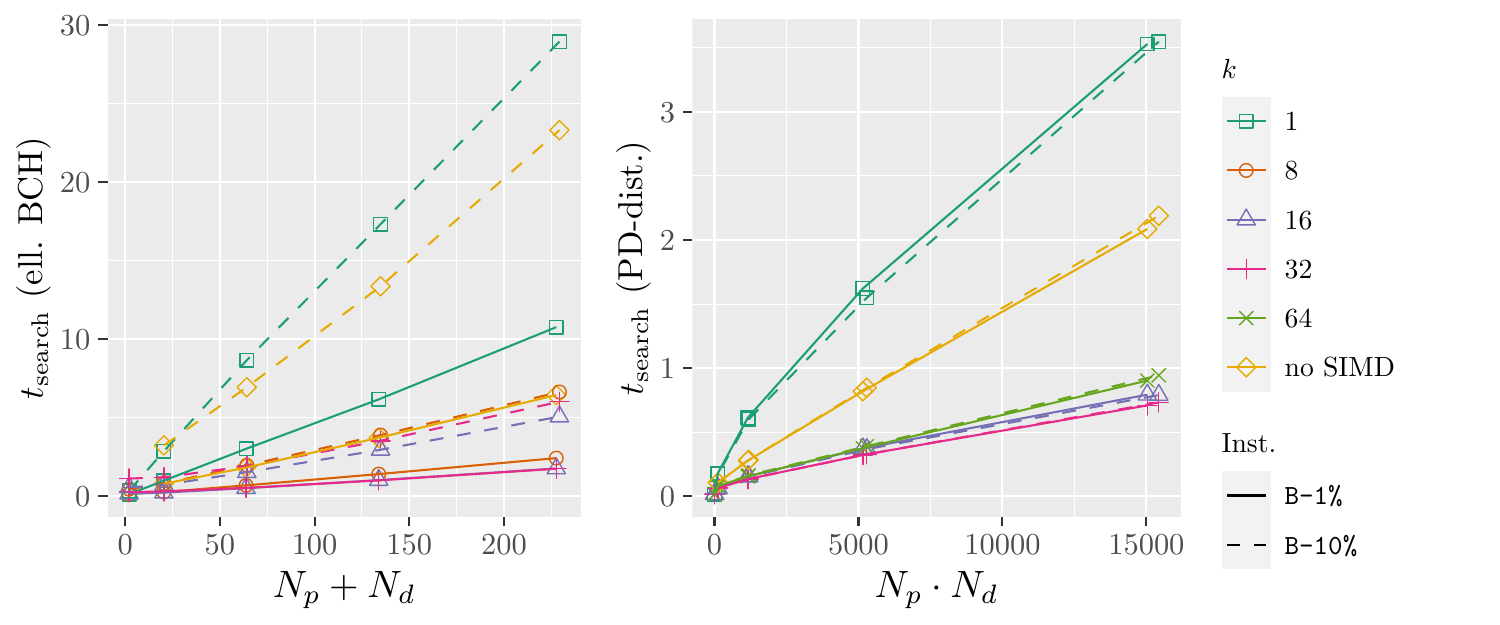}
  \caption{Evaluation of bundled elliptic BCH searches (left) and PD-distance searches (right).
  Shows mean search times (in \ms) for either search on the \BerlinOne and \BerlinTen instances with $\rho \in \{0\s, 150\s, 300\s, 450\s, 600\s\}$.
  Considers $k \in \{1,8,16,32\}$ for elliptic BCH searches and $k \in \{1,16,32,64\}$ for PD-distance searches, using SIMD instructions for $k > 1$.
  Additionally shows running times without SIMD instructions for elliptic searches with $k=16$ and PD-distance searches with $k=32$.
  Note the different $x$- and $y$-axes.}
  \label{fig:ell_bch_and_pd_bundled_eval_plots}
\end{figure}
\begin{figure}[t]
  \centering
  \includeplot{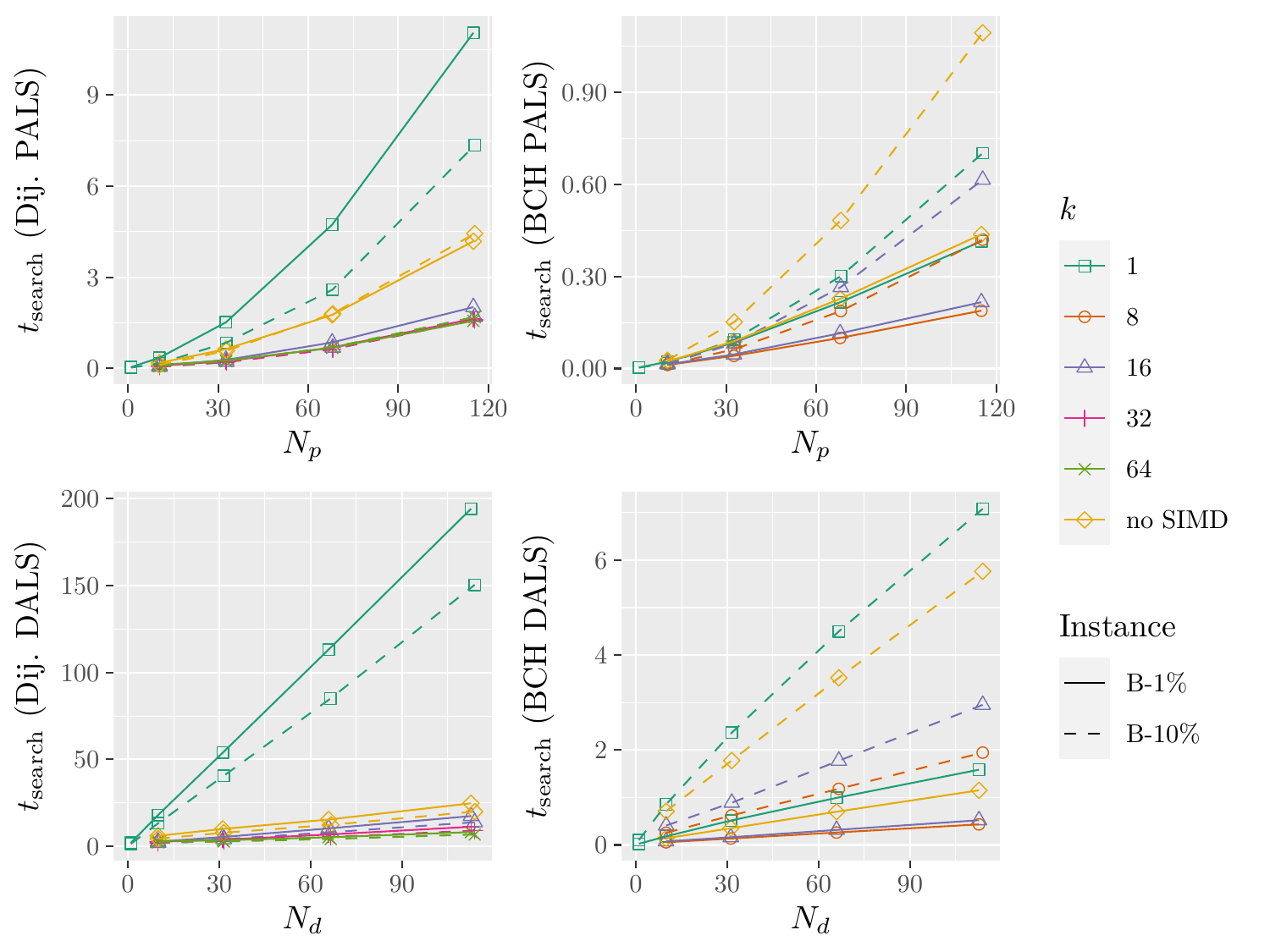}
  \caption{Evaluation of bundled last stop searches.
  Shows mean search times (in \ms) for bundled Dijkstra searches (left) and individual BCH searches (right) in the PALS (top) and DALS (bottom) cases on the \BerlinOne and \BerlinTen instances with $\rho \in \{0\s, 150\s, 300\s, 450\s, 600\s\}$.
  Considers $k \in \{1,16,32,64\}$ for Dijkstra searches and $k \in \{1,8,16\}$ for BCH searches, using SIMD instructions for $k > 1$.
  Additionally shows running times without SIMD instructions for Dijkstra searches with $k=64$ and BCH searches with $k=8$.
  Note the different $y$-axes.}
  \label{fig:last_stop_bundled_eval_plots}
\end{figure}
In this section, we experimentally evaluate bundled searches in each of the described applications and find the optimal value of $k$ for each of them.
We conduct our experiments on the \BerlinOne and \BerlinTen instances with $\rho \in \{0\s, 150\s,300\s,450\s,600\s\}$.
Due to time constraints, we first find the optimal value of $k$ with vector instructions and then evaluate the running time without vector instructions only for that value of $k$.   
We remark, though, that the optimal values of $k$ may differ depending on whether vector instructions are used at all and which SIMD instruction set is available.
We only run a single iteration of Dijkstra searches because of their prohibitive running times.

\parheader{Bundled Elliptic BCH Searches}
We show experimental running times for bundled elliptic BCH searches with $k \in \{ 1,8,16,32 \}$ in~\cref{fig:ell_bch_and_pd_bundled_eval_plots}.
We find that $k=16$ is best suited here.

Because of the use of shortcut edges, the search trees of each source in the CH only start to overlap at larger distances from the sources.
Therefore, we can bundle edge relaxations and bucket scans in the periphery of the sources but not closer to each source.
Sorted buckets additionally shift the focus of the work closer to the sources where large values of $k$ are ineffective.
The value of $k=16$ strikes a balance between the two aspects of limiting the overhead closer to the sources while bundling operations further away.

The advantage of $k=16$ over $k=32$ is more pronounced on the \BerlinTen instance.
We attribute this to the fact that there are more vehicles in total which means that buckets contain more entries on average.
Thus, with sorted buckets, the focus on work closer to the sources is more pronounced. 

Bundled searches with $k=16$ and without SIMD parallelism (\emph{no SIMD}) lead to only small speedups.
Again, bundling works better for the smaller instance due to a larger relative part of the searches' work being performed further away from the sources.

\parheader{Bundled PD-Distance Searches}
In~\cref{fig:ell_bch_and_pd_bundled_eval_plots}, we show the running times for bundled PD-distances searches with $k \in \{1, 16, 32, 64\}$.
Since our PD-distance BCH searches perform more work further away from the sources than elliptic BCH searches, the search trees of all searches overlap more.
Thus, larger values of $k$ allow effective bundling of edge relaxations and the generation and scan of bucket entries.
Nonetheless, the potential for bundling is limited and larger values of $k$ lead to additional overheads closer to the sources.
In effect, $k=32$ is the best choice for our PD-distance searches.

We additionally consider the running time at $k=32$ without SIMD instructions (\emph{no SIMD}).
We observe speed-ups of up to $2.2$ even without SIMD parallelism.
This is again due to the larger amount of work performed in the periphery of the sources that can be bundled well.

\parheader{Bundled Last Stop Searches}
We depict the running times of bundled Dijkstra searches and individual BCH searches for the PALS and DALS cases in~\cref{fig:last_stop_bundled_eval_plots}.

We find that Dijkstra searches are well suited for bundling as we observe the smallest search times with $k=64$.
Since Dijkstra searches do not use shortcut edges, the searches for each individual source meet much earlier than BCH searches.
Thus, the vast majority of the large number of edge relaxations of Dijkstra searches can be bundled well.
This is evidenced by the fact that we see good speedups for bundled Dijkstra searches even without SIMD instructions.
Larger $k \ge 128$ may be useful for larger numbers of sources but eventually we will run into cache limitations as hundreds of bytes of distance labels need to be handled per vertex. 

The bundling of individual BCH searches is faced with the same issue as elliptic BCH searches.
Because of the sorted buckets, most bucket entry scans are performed close to the sources of the BCH queries.
In fact, this issue is even more pronounced for individual BCH queries as the total radius of each search is smaller due to the cost based stopping criterion.
In the PALS case, the radius of each search is so small that bundling without SIMD instructions even at $k=8$ actually increases the running time. 
Bundled searches work better in the DALS case since the stopping criterion is worse which leads to more work at larger distances.
With SIMD parallelism, $k=8$ has almost no drawbacks compared to $k=1$.
Thus, the optimal value in both cases is $k=8$.

\subsection{Sorted Buckets}
\label{subsec:sorted_buckets_experiments}
\begin{figure}[t]
  \centering
  \includeplot{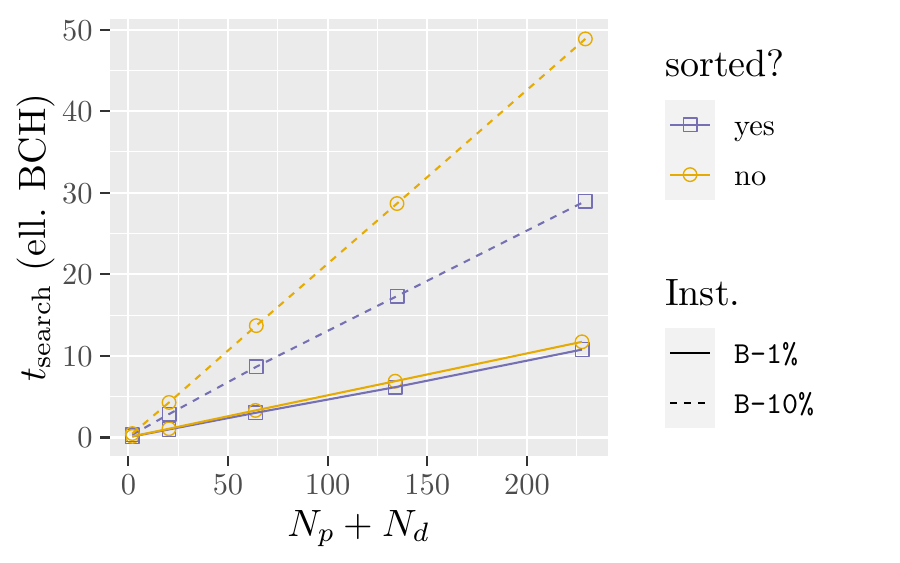}
  \caption{Effectiveness of sorted buckets on running time of elliptic BCH queries.
  Shows mean running times (in \ms) at $k=1$ on the \BerlinOne and \BerlinTen instances for $\rho \in \{ 0\s, 150\s, 300\s, 450\s, 600\s \}$.}
  \label{fig:elliptic_sorted_buckets_eval}
\end{figure}
\begin{figure}[t]
  \centering
  \includeplot{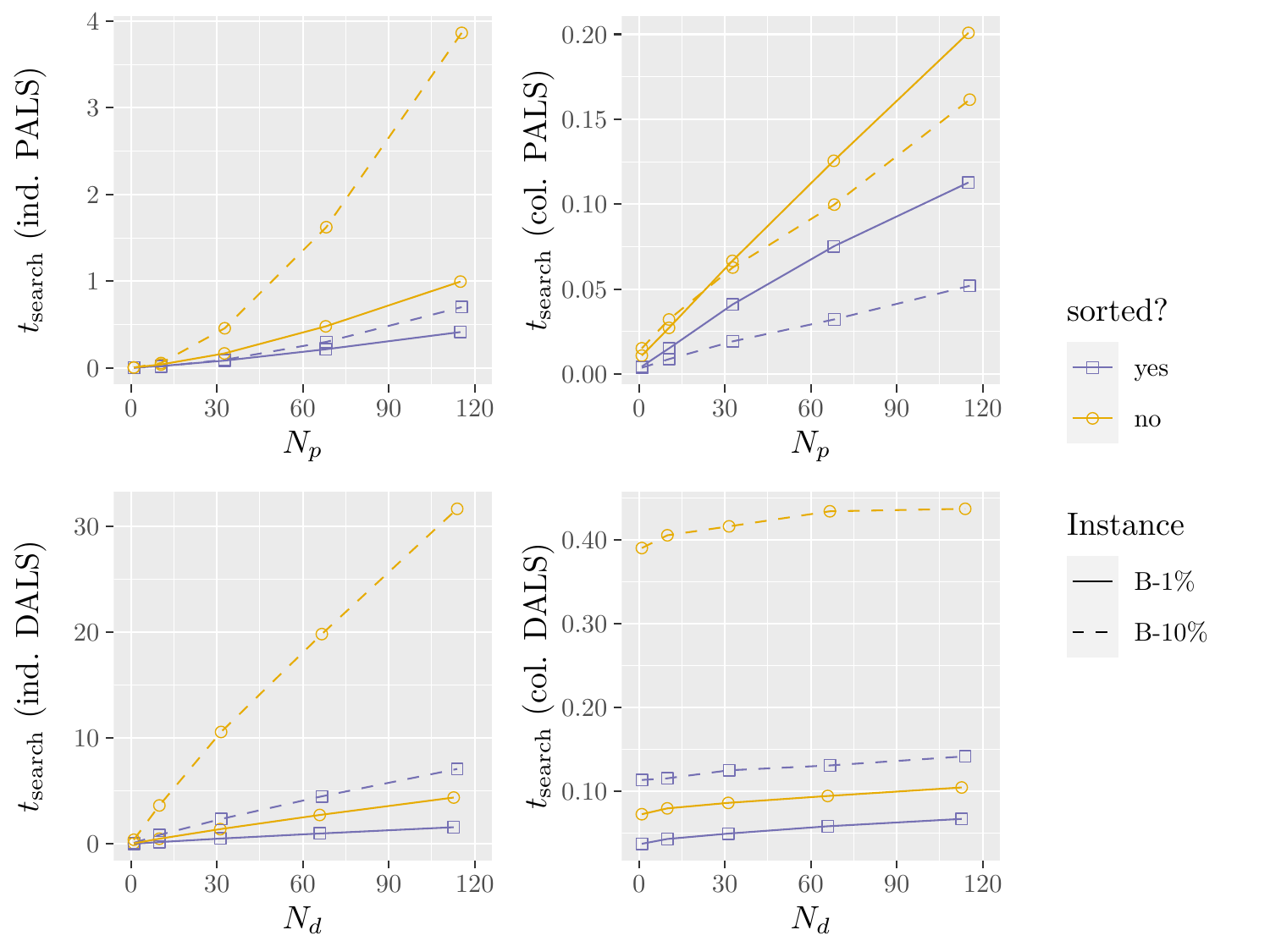}
  \caption{Effectiveness of sorted buckets on running time of last stop BCH queries.
  Shows mean running times (in \ms) of individual BCH queries (left; $k=1$) and collective BCH queries (right) in the PALS (top) and DALS (bottom) cases on the \BerlinOne and \BerlinTen instances for $\rho \in \{ 0\s, 150\s, 300\s, 450\s, 600\s \}$ with and without sorted buckets.
  Note the different $y$-axes.}
  \label{fig:last_stop_sorted_buckets_eval}
\end{figure}
In the following, we analyze the effect of sorted buckets on elliptic BCH searches as well as individual and collective last stop BCH searches.
We consider the reduction in the number of bucket entries scanned as well as the effects on the running time of the searches and the time for updating buckets.
We experimentally compare all searches with sorted and unsorted buckets on the \BerlinOne and \BerlinTen instances with $\rho \in \{ 0\s, 150\s, 300\s, 450\s, 600\s \}$.
For elliptic BCH searches and individual last stop BCH searches, we use $k=1$.
We only run a single iteration with unsorted buckets for elliptic BCH searches and individual last stop BCH searches. 

\parheader{Sorted Buckets for Elliptic BCH Searches}
The buckets for elliptic BCH searches are already strongly pruned using elliptic pruning.
Therefore, sorting these buckets only elicits a major effect with a sufficiently large number of vehicles. 
As we can see in~\cref{fig:elliptic_sorted_buckets_eval}, sorted buckets only have a limited impact for the \BerlinOne instances but a much larger one for the \BerlinTen instance as the latter considers ten times more vehicles.
On the larger input, sorted buckets reduce the number of entries scanned by about half, which leads to a decrease in the search time by up to $40\%$ ($20\ms$).
At the same time, maintaining the order of bucket entries only increases the time for updating bucket entries by less than $10\mus$.
In conclusion, sorted buckets are a valuable improvement for elliptic BCH searches, particularly with respect to the scalability to larger numbers of vehicles. 

\parheader{Sorted Buckets for Last Stop BCH Searches}
For last stop BCH searches, sorted buckets are vital to reduce the number of bucket entries scanned since we cannot use elliptic pruning.
We show the impact of sorted buckets on the last stop search times in~\cref{fig:last_stop_sorted_buckets_eval}.
For individual BCH searches, more than $95\%$ and $75\%$ fewer bucket entries are scanned with sorted buckets in the PALS and DALS cases, respectively.
This reduces search times by up to $82\%$ and $78\%$.
For collective searches, the number of bucket entries scanned reduces by about the same relative factors.
As collective searches scan fewer bucket entries in total, the impact on the search times is less pronounced with reductions of up to $70\%$ in both cases. 

Maintaining sorted last stop buckets incurs an average overhead of less than $20\mus$ per request while the reduction in running time is between one and three orders of magnitude larger.

\subsection{Collective BCH Searches}
\label{subsec:collective_bch_searches_experiments}
\begin{table}[t!]
  \centering
  \caption{
  Comparison of the running times (in \mus) of collective BCH searches (Coll.), individual BCH searches (BCH), and Dijkstra searches (Dij.) for the PALS and DALS case and three radii $\rho \in \{ 0\s, 300\s, 600\s \}$ on the \BerlinOne and \BerlinTen instances.
  Shows number of edge relaxations ($\#_{\text{rel.}}$), number of bucket entries scanned ($\#_{\text{scans}}$), mean search time ($t_{\text{search}}$) and the mean time for enumerating possible insertions ($t_{\text{enum}}$) per request.
  Times marked in bold are the smallest times per radius.
  }
  \setlength{\tabcolsep}{5.5pt}
  \begin{tabular}{|c|R|c|R|R|R|R|R|R|R|R|}
  \cline{4-11}
  \multicolumn{3}{c|}{} & \multicolumn{4}{c|}{\BerlinOne} & \multicolumn{4}{c|}{\BerlinTen} \\
  \cline{2-11}
  \multicolumn{1}{c|}{} & \makecell[c]{\rho} & Alg. & \makecell[c]{\#_{\text{rel.}}} & \makecell[c]{\#_{\text{scans}}} & \makecell[c]{t_{\text{search}}} & \makecell[c]{t_{\text{enum}}} & \makecell[c]{\#_{\text{rel.}}} & \makecell[c]{\#_{\text{scans}}} & \makecell[c]{t_{\text{search}}} & \makecell[c]{t_{\text{enum}}}\\
  \hline 
  \multirowcell{9}{P\\A\\L\\S} & \multirow{3}{*}{0} & Coll.  & 33 & 7 & 4.46 & 0.49 & 12 & 12 & 3.66 & 0.53  \\
  & & BCH  & 30 & 7 & \mathbf{3.16} & 0.37 & 12 & 7 & \mathbf{2.31} & 0.42 \\
  & & Dij.  & 459 & - & 33.26 & \mathbf{0.30} & 179 & - & 14.38 & \mathbf{0.34} \\
  \cline{2-11}
  & \multirow{3}{*}{300} & Coll.  & 225 & 34 & \mathbf{40.98} & \mathbf{0.65} & 68 & 33 & \mathbf{19.47} & \mathbf{0.68} \\
  & & BCH  & 493 & 152 & 41.20 & 27.55 & 357 & 570 & 62.02 & 93.33  \\
  & & Dij. & 2125 & - & 255.21 & 24.64 & 1599 & - & 228.07 & 80.14 \\
  \cline{2-11}
  & \multirow{3}{*}{600} & Coll. & 439 & 64 & \mathbf{112.81} & \mathbf{0.79} & 100 & 52 & \mathbf{52.47} & \mathbf{0.89} \\
  & & BCH & 2243 & 1135 & 191.32 & 408.40 & 1661 & 5530 & 420.25 & 1839.26 \\
  & & Dij. & 14821 & - & 1551.77 & 366.51 & 12852 & - & 1689.69 & 1681.56  \\
  \hline
  \multirowcell{9}{D\\A\\L\\S} & \multirow{3}{*}{0} & Coll.   & 177 & 1077 & 36.31 & 3.15 & 154 & 7895 & \mathbf{112.76} & \mathbf{6.89}  \\
  & & BCH  & 183 & 1128 & \mathbf{26.14} & \mathbf{3.12} & 159 & 8283 & 116.17 & 7.27 \\
  & & Dij.  & 22317 & - & 1978.23 & 28.55 & 15780 & - & 1411.81 & 132.11 \\
  \cline{2-11}
  & \multirow{3}{*}{300} & Coll.  & 197 & 1053 & \mathbf{48.74} & \mathbf{7.46} & 173 & 7857 & \mathbf{126.67} & \mathbf{20.13} \\
  & & BCH  & 1132 & 4792 & 134.84 & 97.13 & 1035 & 37684 & 605.15 & 363.65  \\
  & & Dij. & 24762 & - & 3444.91 & 174.46 & 18522 & - & 2667.37 & 694.29 \\
  \cline{2-11}
  & \multirow{3}{*}{600} & Coll. & 213 & 1043 & \mathbf{66.51} & \mathbf{14.84} & 189 & 7837 & \mathbf{141.38} & \mathbf{40.57} \\
  & & BCH & 3868 & 15484 & 432.85 & 754.17 & 3509 & 122199 & 1967.98 & 2403.68 \\
  & & Dij. & 62087 & - & 8149.34 & 1639.57 & 49489 & - & 6807.04 & 5992.58  \\
  \hline
  \end{tabular}
  \setlength{\tabcolsep}{6pt}
  \label{tab:last_stop_search_comparison}
\end{table}
In~\cref{tab:last_stop_search_comparison}, we compare the search times and the times needed for the enumeration of candidate insertions for the three search approaches used for the PALS and DALS cases.
Additionally, we show the number of relaxed edges and scanned bucket entries.
We show the results for $\rho \in \{ 0\s, 300\s, 600\s \}$ on the \BerlinOne and \BerlinTen instances.
We only conducted one iteration of the experiments for Dijkstra searches because of their large running time. 

We find that collective searches are slower than individual BCH searches at $\rho=0\s$ (except for the DALS case on \ShortBerlinOne).
This is due to the fact that there is only a single pickup and dropoff which means the overhead for propagating labels instead of only distances is unwarranted.

At $\rho=300\s$ and $\rho=600\s$, collective searches offer the best search times and by far the best enumeration times, though. 
The search times of collective searches are up to an order of magnitude smaller than those of individual BCH searches.
We attribute this to two main advantages of collective searches.

Firstly, collective searches can be pruned more precisely because we use lower bounds on the cost of specific PD-pairs or dropoffs instead of a general lower bound on the cost of every PD-pair or dropoff.
This applies to the stopping criteria for bucket scans and for the searches as a whole.

Secondly, collective searches consider all sources in one search, maximizing the amount of information that can be used by domination pruning.
Bundled searches can only consider $k$ searches at once with time overheads for $k > 8$ (s.a.). 
Therefore, each edge may be scanned by up to $\numpickups / 8$ or $\numdropoffs / 8$ bundled searches with no way to bundle relaxations between searches.
Thus, the number of edge relaxations and bucket entry scans increases much faster with the number of PD-locations ($ \numpickups, \numdropoffs \sim \rho^2$) for individual BCH searches than for collective BCH searches.
In fact, domination pruning works so well in the DALS case that the search time is virtually constant with an increasing number of dropoffs (cf.~\cref{fig:last_stop_sorted_buckets_eval}).

In addition, the enumeration times for collective searches are almost constant, too, while they increase massively with $\rho$ for individual BCH searches.
This is due to the fact that collective searches identify a single candidate insertion during the search while individual BCH searches first find all distances and then enumerate an insertion for each combination of candidate vehicle and PD-pair.
Since the number of PD-pairs is proportional to $\rho^4$, enumeration times quickly become very large with tens of thousands of insertions tried.

\mysubsection{Comparison with LOUD}
\label{subsec:comparison_with_LOUD}
\begin{table}[t]
  \centering
  \caption{Running times (in \mus) of different phases of LOUD ($\rho=\text{L-}\dots$) and \karri with different radii ($\rho \in \{0\s,300\s,600\s\}$) on \ShortBerlinOne, \ShortBerlinTen, \ShortRuhrOne, and \ShortRuhrTen.
  Shows mean times for finding $\Prho$ and $\Drho$, PD-distance searches, elliptic BCH searches, enumerating ordinary and PBNS insertions, PALS and DALS searches, and updating routes and buckets as well as the mean total time per request.
  Numbers with asterisks are estimates.
  }
  \begin{tabular}{|c|R|R|R|R|R|R|R|R|R|}
    \hline
    \makecell[c]{Inst.} & \makecell[c]{\rho} & \makecell[c]{\text{find} \\ \Prho, \Drho} & \makecell[c]{\text{PD}} & \makecell[c]{\text{BCH}} & \makecell[c]{\text{Ord.\&}\\\text{PBNS}} & \makecell[c]{\text{PALS}} & \makecell[c]{\text{DALS}} & \makecell[c]{\text{update}} & \makecell[c]{\text{total}} \\
    \hline
    \multirowcell{4}{\ShortBerlinOne} & \text{L-} 0 & 0 & 17 & 124 & 147 & 43 & 1974 & 84 & 2388 \\
    & 0 & 2 & 36 & 123 & 42 & 3 & 28 & 132 & 367 \\
    & 300 & 95 & 155 & 497 & 119 & 43 & 56 & 138 & 1103 \\
    & 600 & 316 & 735 & 1703 & 469 & 156 & 81 & 140 & 3601 \\
    \hline
    \multirowcell{6}{\ShortBerlinTen} & \text{L-} 0 & 0 & 16 & 368 & 415 & 23 & 1474 & 80 & 2376 \\
    & 0 & 3 & 35 & 362 & 353 & 3 & 119 & 200 & 1074 \\
    & \text{L-} 300 & ^\ast 104 & ^\ast 18930 & 28675 & 736 & 918 & 44071 & ^\ast 80 & ^\ast 93514 \\
    & 300 & 104 & 157 & 1502 & 641 & 24 & 140 & 203 & 2772 \\
    & \text{L-} 600 & ^\ast 353 & ^\ast 246521 & 48856 & 1900 & 8998 & 166244 & ^\ast 80 & ^\ast 472952 \\
    & 600 & 353 & 759 & 4804 & 1741 & 106 & 181 & 213 & 8158 \\
    \hline
    \multirowcell{4}{\ShortRuhrOne} & \text{L-} 0 & 0 & 19 & 173 & 228 & 128 & 8572 & 82 & 9202 \\
    & 0 & 3 & 18 & 163 & 106 & 5 & 34 & 138 & 481 \\
    & 300 & 98 & 132 & 668 & 194 & 62 & 59 & 146 & 1375 \\
    & 600 & 293 & 566 & 1989 & 494 & 216 & 81 & 149 & 3807 \\
    \hline
    \multirowcell{4}{\ShortRuhrTen} & \text{L-} 0 & 0 & 18 & 703 & 944 & 89 & 6154 & 80 & 7988 \\
    & 0 & 3 & 19 & 708 & 944 & 5 & 163 & 298 & 2155 \\
    & 300 & 106 & 138 & 3360 & 1293 & 42 & 182 & 320 & 5462 \\
    & 600 & 324 & 595 & 9396 & 2476 & 128 & 226 & 387 & 13554 \\
    \hline
  \end{tabular}
  \label{tab:comparison_with_LOUD_run_times}
\end{table}
\begin{table}[t]
  \centering
  \caption{Solution quality of \karri with different radii ($\rho \in \{0\s,300\s,600\s\}$) on \ShortBerlinOne, \ShortBerlinTen, \ShortRuhrOne, and \ShortRuhrTen.
  For requests, we report the average and 95\%-quantile wait time, and the average ride and trip times (in mm:ss).
  For vehicles, we give the average times spent driving empty, driving occupied, and making stops, as well as the average total operation time (in hh:mm).}
  \begin{tabular}{|c|r|r|r|r|r|r|r|r|r|}
    \hline
    \makecell[c]{Inst.} & \makecell[c]{$\rho$} & \makecell[c]{\text{wait}} & \makecell[c]{\text{w.-95\%}} & \makecell[c]{\text{ride}} & \makecell[c]{\text{trip}}& \makecell[c]{\text{empty}} & \makecell[c]{\text{occ}} & \makecell[c]{\text{stop}} & \makecell[c]{\text{op}} \\
    \hline
    \multirowcell{3}{\ShortBerlinOne} & 0 & 3:49 & 9:38 & 12:32 & 17:18 & 0:41 & 3:10 & 0:28 & 4:19 \\
    & 300 & 3:22 & 8:24 & 11:56 & 16:13 & 0:30 & 2:56 & 0:30 & 3:56 \\
    & 600 & 3:27 & 8:34 & 11:47 & 16:11 & 0:29 & 2:51 & 0:30 & 3:51 \\
    \hline
    \multirowcell{3}{\ShortBerlinTen} & 0 & 2:31 & 7:11 & 11:57 & 15:07 & 0:16 & 2:29 & 0:26 & 3:10 \\
    & 300 & 2:20 & 6:20 & 11:33 & 14:41 & 0:09 & 2:12 & 0:27 & 2:48 \\
    & 600 & 2:30 & 7:34 & 11:27 & 14:50 & 0:08 & 2:08 & 0:27 & 2:42 \\
    \hline
    \multirowcell{3}{\ShortRuhrOne} & 0 & 4:49 & 11:48 & 12:25 & 18:26 & 0:56 & 3:15 & 0:27 & 4:39 \\
    & 300 & 4:15 & 10:20 & 11:48 & 17:07 & 0:44 & 3:03 & 0:29 & 4:16 \\
    & 600 & 4:17 & 9:52 & 11:33 & 16:52 & 0:42 & 2:58 & 0:29 & 4:09 \\
    \hline
    \multirowcell{3}{\ShortRuhrTen} & 0 & 3:11 & 8:46 & 11:46 & 15:48 & 0:24 & 2:40 & 0:25 & 3:29 \\
    & 300 & 2:45 & 7:07 & 11:13 & 14:51 & 0:16 & 2:26 & 0:26 & 3:08 \\
    & 600 & 2:55 & 7:58 & 11:06 & 14:56 & 0:15 & 2:22 & 0:26 & 3:03 \\
    \hline
  \end{tabular}
  \label{tab:comparison_with_LOUD_solution_quality}
\end{table}
In this section, we compare our approach with the LOUD algorithm~\cite{buchhold2021fast}.

\parheader{Running Times} 
We give the running times for the different phases of both algorithms on the \BerlinOne, \BerlinTen, \RuhrOne, and \RuhrTen instances in~\cref{tab:comparison_with_LOUD_run_times}.
For \karri, we consider $\rho \in \{ 0\s, 300\s, 600\s \}$ and use the optimal last stop search approach in each configuration.
We also report the running times of the LOUD algorithm ($\rho=\text{L-}0)$) on all four instances.
Additionally, we consider an estimate for the running time of a na\"ive extension of LOUD to multiple PD-locations ($\rho=\text{L-}300$ and $\rho=\text{L-}600$) for the medium sized \BerlinTen instance.

First, we consider the scenario with a single pickup and dropoff ($\rho=0\s$) which is the scenario considered by LOUD.
Here, sorted buckets have little impact on the search times of elliptic BCH searches even though the number of bucket entries scanned is reduced.
We attribute this to the fact that our implementation is meant to deal with any number of PD-locations while LOUD is specialized for the case of $\numpickups=\numdropoffs=1$.
Our last stop BCH searches are well suited for $\rho=0\s$, though.
They are up to $27$ and $250$ times faster than LOUD's Dijkstra searches in the PALS and DALS cases, respectively.
Maintaining sorted last stop buckets does lead to increased update times, though, especially for the larger instance where buckets contain more entries.
In total, we can reduce the average time per request by factors of $6.5$ for \BerlinOne and $2.2$ for \BerlinTen compared to LOUD.
On the larger Ruhr instances, we achieve speedups of $19.1$ and $3.7$ for the \RuhrOne and \RuhrTen instances.

Next, we consider our estimate for an extension of LOUD that uses only the techniques of the original algorithm.
For this, we configured \karri to use no bundled searches, no sorted buckets, and to use Dijkstra searches for the PALS and DALS cases.
For the PD-distances, we obtained an estimate for the time of running one CH-query per PD-pair by multiplying LOUD's PD-distance search time with $\numpickups \cdot \numdropoffs$.
We find that bundling and sorted buckets make elliptic BCH searches about one order of magnitude faster than the na\"ive extension.
For the PALS and DALS searches, our collective BCH approach beats the standard Dijkstra approach by two and three orders of magnitude, respectively.
We assume that the CH-queries for PD-distances would in reality be faster than our estimate.
Nonetheless, it is notable that our BCH based approach is hundreds of times faster than the estimate. 

\parheader{Solution Quality}
In the following, we give a first idea of how trip times and vehicle operation times can be improved by extending ridesharing with walking.

In~\cref{tab:comparison_with_LOUD_solution_quality}, we compare the solution quality of \karri with $\rho \in \{ 0\s, 300\s, 600\s \}$.
Note that we allow passengers to walk to their destinations if the resulting walking time leads to better cost than a ridesharing trip (see~\cref{subsec:walking_time_and_walking_to_the_destination}).
The cost function used here equally weights the passenger trip times and vehicle operation times ($\tripweight=1$).
With larger values of $\rho$, we observe improvements in both criteria.
At $\rho=300\s$, the average vehicle operations times and passenger wait times improve by up to $12\%$ while trip times improve by up to $7\%$.
At $\rho=600\s$, the vehicle operation times and passenger ride times reduce further compared to $\rho=300\s$.

There are a number of parameters not evaluated in these preliminary results.
For instance, we have not considered the willingness of passengers to walk longer distances in order to reduce trip times in our cost function (since $\walkweight = 0$).
Also, we only show results for a fixed number of vehicles and a density of requests representing a regular week day.
In the future, we would like to include these parameters in our analysis of the effects of PD-locations on ridesharing and extend this evaluation to larger inputs.

\section{Conclusions and Future Work}
\label{sec:conclusion}
\karri develops efficient many-to-many routing
with bucket contraction hierarchies that allows 
efficient scheduling of large vehicle fleets
considering many pickup and dropoff locations.  
A flexible cost function allows
configuration to many situations, e.g. using
walking, bicycles or scooters.  We expect that the new
techniques like sorted buckets can also be applied for other problems
that use many-to-many routing with correlated sources
and targets.

Next, we want to use \karri to evaluate different
traffic scenarios.  For example, we expect that
even larger savings over single PD-locations are
possible when using faster individual transport
or when using fewer but larger vehicles.

Future work on the algorithmic side can achieve acceleration by
clustering routing sources into batches of size $k$ by
their proximity, by finding a collective approach for
elliptic BCH searches, and by parallelizing the algorithm
both over requests (that use different vehicles) and over
different PD-locations.
The overall cost could be further optimized by
going away from greedy online scheduling, 
instead taking into account pre-booked trips and
opportunities to transparently change existing
trips for local search style optimizations.

We expect that additional generalizations can
integrate \karri with public transportation where
pickup and dropoff locations can be stops of
buses or trains and where the cost function has
to take into account the public transportation
schedule.

A longer term perspective is to allow switching
vehicles during a trip.  This opens up the possibility
of more sharing using larger vehicles,
eventually leading to a highly adaptive software defined public transportation system.  This 
implies interesting algorithmic challenges as it
leads to a combinatorial explosion
of possible route options.



\bibliography{karri}

\begin{thebibliography}{10}

\bibitem{bauer2010sharc}
Reinhard Bauer and Daniel Delling.
\newblock {SHARC}: {F}ast and {R}obust {U}nidirectional {R}outing.
\newblock {\em ACM Journal of Experimental Algorithms (JEA)}, 14, 2010.
\newblock \href {https://doi.org/10.1145/1498698.1537599}
  {\path{doi:10.1145/1498698.1537599}}.

\bibitem{8317926}
Joschka Bischoff, Michal Maciejewski, and Kai Nagel.
\newblock {C}ity-{W}ide {S}hared {T}axis: {A} {S}imulation {S}tudy in {B}erlin.
\newblock In {\em IEEE 20th International Conference on Intelligent
  Transportation Systems (ITSC)}, 2017.
\newblock \href {https://doi.org/10.1109/ITSC.2017.8317926}
  {\path{doi:10.1109/ITSC.2017.8317926}}.

\bibitem{buchhold2019real}
Valentin Buchhold, Peter Sanders, and Dorothea Wagner.
\newblock {R}eal-{T}ime {T}raffic {A}ssignment using {E}ngineered
  {C}ustomizable {C}ontraction {H}ierarchies.
\newblock {\em ACM Journal of Experimental Algorithms (JEA)}, 24, 2019.
\newblock \href {https://doi.org/10.1145/3362693} {\path{doi:10.1145/3362693}}.

\bibitem{buchhold2021fast}
Valentin Buchhold, Peter Sanders, and Dorothea Wagner.
\newblock {F}ast, {E}xact and {S}calable {D}ynamic {R}idesharing.
\newblock In {\em SIAM Symposium on Algorithm Engineering and Experiments
  (ALENEX)}, 2021.
\newblock \href {https://doi.org/10.1137/1.9781611976472.8}
  {\path{doi:10.1137/1.9781611976472.8}}.

\bibitem{delling2013phast}
Daniel Delling, Andrew~V Goldberg, Andreas Nowatzyk, and Renato~F Werneck.
\newblock {PHAST}: {H}ardware-{A}ccelerated {S}hortest {P}ath {T}rees.
\newblock {\em Journal of Parallel and Distributed Computing}, 73, 2013.
\newblock \href {https://doi.org/10.1016/j.jpdc.2012.02.007}
  {\path{doi:10.1016/j.jpdc.2012.02.007}}.

\bibitem{delling2017customizable}
Daniel Delling, Andrew~V Goldberg, Thomas Pajor, and Renato~F Werneck.
\newblock {C}ustomizable {R}oute {P}lanning in {R}oad {N}etworks.
\newblock {\em INFORMS Transportation Science}, 51(2), 2017.
\newblock \href {https://doi.org/10.1287/trsc.2014.0579}
  {\path{doi:10.1287/trsc.2014.0579}}.

\bibitem{delling2011faster}
Daniel Delling, Andrew~V Goldberg, and Renato~F Werneck.
\newblock {F}aster {B}atched {S}hortest {P}aths in {R}oad {N}etworks.
\newblock In {\em 11th Workshop on Algorithmic Approaches for Transportation
  Modelling, Optimization, and Systems (ATMOS)}. LIPIcs, 2011.
\newblock \href {https://doi.org/10.4230/OASIcs.ATMOS.2011.52}
  {\path{doi:10.4230/OASIcs.ATMOS.2011.52}}.

\bibitem{dijkstra1959note}
Edsger~W Dijkstra.
\newblock {A} {N}ote on {T}wo {P}roblems in {C}onnexion with {G}raphs.
\newblock {\em Numerische Mathematik}, 1, 1959.

\bibitem{geisberger2012exact}
Robert Geisberger, Peter Sanders, Dominik Schultes, and Christian Vetter.
\newblock {E}xact {R}outing in {L}arge {R}oad {N}etworks using {C}ontraction
  {H}ierarchies.
\newblock {\em INFORMS Transportation Science}, 46(3), 2012.
\newblock \href {https://doi.org/10.1287/trsc.1110.0401}
  {\path{doi:10.1287/trsc.1110.0401}}.

\bibitem{7549045}
Preeti Goel, Lars Kulik, and Kotagiri Ramamohanarao.
\newblock {O}ptimal {P}ick up {P}oint {S}election for {E}ffective {R}ide
  {S}haring.
\newblock {\em IEEE Transactions on Big Data}, 3(2), 2017.
\newblock \href {https://doi.org/10.1109/TBDATA.2016.2599936}
  {\path{doi:10.1109/TBDATA.2016.2599936}}.

\bibitem{hilger2009fast}
Moritz Hilger, Ekkehard K{\"o}hler, Rolf~H M{\"o}hring, and Heiko Schilling.
\newblock {F}ast {P}oint-to-{P}oint {S}hortest {P}ath {C}omputations with
  {A}rc-{F}lags.
\newblock {\em The Shortest Path Problem: 9th DIMACS Implementation Challenge},
  2009.

\bibitem{horni2016multi}
Andreas Horni, Kai Nagel, and Kay Axhausen, editors.
\newblock {\em {T}he {M}ulti-{A}gent {T}ransport {S}imulation {MATSim}}.
\newblock Ubiquity Press, 2016.
\newblock \href {https://doi.org/10.5334/baw} {\path{doi:10.5334/baw}}.

\bibitem{knopp2007computing}
Sebastian Knopp, Peter Sanders, Dominik Schultes, Frank Schulz, and Dorothea
  Wagner.
\newblock {C}omputing {M}any-to-{M}any {S}hortest {P}aths using {H}ighway
  {H}ierarchies.
\newblock In {\em SIAM Workshop on Algorithm Engineering and Experiments
  (ALENEX)}, 2007.
\newblock \href {https://doi.org/10.1137/1.9781611972870.4}
  {\path{doi:10.1137/1.9781611972870.4}}.

\bibitem{Morency2007}
Catherine Morency.
\newblock {T}he {A}mbivalence of {R}idesharing.
\newblock {\em Transportation}, 34(2), 2007.
\newblock \href {https://doi.org/10.1007/s11116-006-9101-9}
  {\path{doi:10.1007/s11116-006-9101-9}}.

\bibitem{Song2021}
Changle Song, Julien Monteil, Jean-Luc Ygnace, and David Rey.
\newblock {I}ncentives for {R}idesharing: {A} {C}ase {S}tudy of {W}elfare and
  {T}raffic {C}ongestion.
\newblock {\em Journal of Advanced Transportation}, 2021.
\newblock \href {https://doi.org/10.1155/2021/6627660}
  {\path{doi:10.1155/2021/6627660}}.

\bibitem{STIGLIC201536}
Mitja Stiglic, Niels Agatz, Martin Savelsbergh, and Mirko Gradisar.
\newblock {T}he {B}enefits of {M}eeting {P}oints in {R}ide-{S}haring {S}ystems.
\newblock {\em Transportation Research Part B: Methodological}, 82, 2015.
\newblock \href {https://doi.org/10.1016/j.trb.2015.07.025}
  {\path{doi:10.1016/j.trb.2015.07.025}}.

\bibitem{yanagisawa2010multi}
Hiroki Yanagisawa.
\newblock {A} {M}ulti-{S}ource {L}abel-{C}orrecting {A}lgorithm for the
  {A}ll-{P}airs {S}hortest {P}aths {P}roblem.
\newblock In {\em IEEE International Symposium on Parallel \& Distributed
  Processing (IPDPS)}, 2010.
\newblock \href {https://doi.org/10.1109/IPDPS.2010.5470362}
  {\path{doi:10.1109/IPDPS.2010.5470362}}.

\bibitem{YU2017141}
Biying Yu, Ye~Ma, Meimei Xue, Baojun Tang, Bin Wang, Jinyue Yan, and Yi-Ming
  Wei.
\newblock {E}nvironmental {B}enefits from {R}idesharing: {A} {C}ase of
  {B}eijing.
\newblock {\em Applied Energy}, 191, 2017.
\newblock \href {https://doi.org/10.1016/j.apenergy.2017.01.052}
  {\path{doi:10.1016/j.apenergy.2017.01.052}}.

\bibitem{ziemke2019matsim}
Dominik Ziemke, Ihab Kaddoura, and Kai Nagel.
\newblock {T}he {MATSim} {O}pen {B}erlin {S}cenario: {A} {M}ultimodal
  {A}gent-{B}ased {T}ransport {S}imulation {S}cenario {B}ased on {S}ynthetic
  {D}emand {M}odeling and {O}pen {D}ata.
\newblock {\em Procedia Computer Science}, 151, 2019.
\newblock \href {https://doi.org/10.1016/j.procs.2019.04.120}
  {\path{doi:10.1016/j.procs.2019.04.120}}.

\end{thebibliography}


\end{document}